\documentclass[11pt]{article}
\usepackage{longtable}
\usepackage{caption}
\usepackage{multirow,multicol}
\usepackage{epsf, graphicx}
\usepackage{latexsym,amsfonts,amsbsy,amssymb}
\usepackage{amsmath,amsthm}
\usepackage{cite}
\usepackage{cleveref}
\usepackage{indentfirst}
\usepackage{bm}
\usepackage{enumitem}
\setenumerate[1]{itemsep=0pt,partopsep=0pt,parsep=\parskip,topsep=5pt}
\setitemize[1]{itemsep=0pt,partopsep=0pt,parsep=\parskip,topsep=5pt}
\setdescription{itemsep=0pt,partopsep=0pt,parsep=\parskip,topsep=5pt}
\allowdisplaybreaks[4]
\makeatletter

\@addtoreset{equation}{section} \makeatother
\newtheorem{Theorem}{Theorem}[section]
\newtheorem{Lemma}{Lemma}[section]

\newtheorem{Definition}{Definition}[section]

\newtheorem{Example}{Example}[section]
\makeatletter \@addtoreset{equation}{section} \makeatother
\begin{document}
	
	\title{The dual codes of two classes of LCD BCH codes \let\thefootnote\relax\footnotetext{E-Mail addresses: yuqingfu@mails.ccnu.edu.cn (Y. Fu), hwliu@ccnu.edu.cn (H. Liu).}}
	\author{ Yuqing Fu,~ Hongwei Liu}
	\date{\small School of Mathematics and Statistics, Central China Normal University, Wuhan, 430079, China}
	\maketitle
	{\noindent\small{\bf Abstract:} Cyclic BCH codes and negacyclic BCH codes form important subclasses of cyclic codes and negacyclic codes, respectively, and can produce optimal linear codes in many cases. To the best of our knowledge, there are few results on the dual codes of cyclic and negacyclic BCH codes. In this paper, we study the dual codes of narrow-sense cyclic BCH codes of length $q^m+1$ over a finite field $\mathbb{F}_q$, where $q$ is an odd prime power, and the dual codes of narrow-sense negacyclic BCH codes of length $\frac{q^{m}+1}{2}$ over $\mathbb{F}_q$, where $q$ is an odd prime power satisfying $q\equiv 3~({\rm mod}~4)$. Some lower bounds on the minimum distances of the dual codes are established, which are very close to the true minimum distances of the dual codes in many cases. Sufficient and necessary conditions  for the even-like subcodes of narrow-sense cyclic BCH codes of length $q^{m}+1$ being cyclic dually-BCH codes are given in terms of designed distances, where $q$ is odd and $m$ is odd or $m\equiv 2~({\rm mod~}4)$. The concept of negacyclic dually-BCH codes is proposed, and sufficient and necessary conditions in terms of designed distances are presented to ensure that narrow-sense negacyclic BCH codes of length $\frac{q^{m}+1}{2}$ are dually-BCH codes, where $q\equiv 3~({\rm mod}~4)$.}
	
	\vspace{1ex}
	
	\section{Introduction}
	Let $q$ be a prime power and $\mathbb{F}_{q}$ be the finite field with $q$ elements. An $[n,k,d]$ linear code 
	$\mathcal{C}$ over $\mathbb{F}_{q}$ is a $k$-dimensional subspace of $\mathbb{F}_{q}^{n}$ with minimum distance $d$. The dual code of $\mathcal{C}$, denoted by $\mathcal{C}^{\perp}$, is defined as 
	$$\mathcal{C}^{\perp}=\big\{{\bf x}\in \mathbb{F}_{q}^{n}~\big|~{\bf x}\cdot {\bf c}=0~{\rm for~all}~{\bf c}\in \mathcal{C}\big\},$$
	where ${\bf x}\cdot {\bf c}$ is the standard inner product of the two vectors ${\bf x}$ and ${\bf c}$. Denote by $\mathbb{F}_q^{*}$ the multiplicative group of $\mathbb{F}_q$ and take $\lambda\in \mathbb{F}_q^{*}$. A linear code $\mathcal{C}$ is said to be $\lambda$-constacyclic if $(c_{0},c_{1},\cdots,c_{n-1})\in \mathcal{C}$ implies $(\lambda c_{n-1},c_{0},\cdots,c_{n-2})\in \mathcal{C}$. Such a code is called cyclic if $\lambda=1$, and negacyclic if $\lambda=-1$. As usual, $\mathbb{F}_{q}[x]$ denotes the polynomial ring over $\mathbb{F}_{q}$ and $\mathcal{R}_{n,\lambda}:=\mathbb{F}_{q}[x]/\langle x^{n}-\lambda\rangle$ denotes the quotient ring of $\mathbb{F}_{q}[x]$ with respect to the ideal $\langle x^{n}-\lambda\rangle$. By identifying each vector $(a_{0},a_{1},\cdots,a_{n-1})\in \mathbb{F}_{q}^{n}$ with a polynomial $$a_{0}+a_{1}x+\cdots+a_{n-1}x^{n-1}\in \mathcal{R}_{n,\lambda},$$ a linear code $\mathcal{C}$ of length $n$ over $\mathbb{F}_{q}$ corresponds to an $\mathbb{F}_{q}$-subspace of the ring $\mathcal{R}_{n,\lambda}$. Moreover, $\mathcal{C}$ is $\lambda$-constacyclic if and only if the corresponding subspace is an ideal of $\mathcal{R}_{n,\lambda}$. Let $\mathcal{C}$ be a $\lambda$-constacyclic code of length $n$ over $\mathbb{F}_{q}$, then $\mathcal{C}=\langle g(x)\rangle$, where $g(x)$ is a monic divisor of $x^{n}-\lambda$ with minimum degree among all the generators of $\mathcal{C}$. Such polynomial $g(x)$ is unique and is called the generator polynomial of $\mathcal{C}$. 
	
	Let $n$ be a positive integer coprime to $q$. Denote by $r$ the order of $\lambda$ in $\mathbb{F}_q^{*}$. There exists a primitive $rn$-th root of unity $\beta$ in some extension field $\mathbb{F}$ of $\mathbb{F}_q$ such that $\beta^n=\lambda$. Let $\mathbb{Z}_n$ be the residue ring of the integer ring $\mathbb{Z}$ modulo $n$. The set of roots of $x^n-\lambda$ in $\mathbb{F}$ corresponds to a subset $1+r\mathbb{Z}_{rn}$ of the residue ring $\mathbb{Z}_{rn}$, which is defined as follows:
	$$1+r\mathbb{Z}_{rn}=\{1+ri\mid 0\leq i\leq n-1\}.$$
	And,
	$$x^n-\lambda=\prod_{j\in 1+r\mathbb{Z}_{rn}}(x-\beta^j).$$
	The set $T=\{j\in 1+r\mathbb{Z}_{rn}~|~g(\beta^{j})=0\}$ is referred to as the defining set of $\mathcal{C}$ with respect to $\beta$. For integers $b=1+ri'$ and $\delta$ with $i'\in \mathbb{Z}$ and $2\leq \delta \leq n$, let $\mathcal{C}_{(q,n,\lambda, \delta,b)}$ denote the $\lambda$-constacyclic code of length $n$ over $\mathbb{F}_{q}$ with defining set 
	$$T=\cup_{i=0}^{\delta-2}C_{b+ri}^{(q,rn)},\eqno{(1)}$$
	where $C_{b+ri}^{(q,rn)}$ is the $q$-cyclotomic coset modulo $rn$ containing $b+ri$. The code $\mathcal{C}_{(q,n,\lambda,\delta,b)}$ is called a $\lambda$-constacyclic BCH code with designed distance $\delta$ with respect to $\beta$. If $b=1$, $\mathcal{C}_{(q,n,\lambda,\delta,b)}$ is called narrow-sense, and non-narrow-sense, otherwise. For narrow-sense cyclic BCH code $\mathcal{C}_{(q,n,1,\delta,1)}$, the code $\mathcal{C}_{(q,n,1,\delta+1,0)}$ is exactly the even-like subcode of $\mathcal{C}_{(q,n,1,\delta,1)}$.
	
	
	The most extensively studied types of constacyclic BCH codes are cyclic BCH codes, which form a special class of cyclic codes and are widely used in compact discs, digital audio tapes and other data storage systems to improve data reliability. In the past several decades, the parameters of cyclic BCH codes including the dimensions, the minimum distances and the weight distributions have been the hot issues of research. Until now, we have very limited knowledge about the duals of cyclic BCH codes over $\mathbb{F}_q$. In \cite{gdl} and \cite{wwlw}, the authors showed several lower bounds on the minimum distances of the duals of narrow-sense cyclic BCH codes of lengths $q^m-1$, $\frac{q^m-1}{q-1}$ and $\frac{q^m-1}{q+1}$, respectively. Wang et al. \cite{wxz} generalized the results in \cite{gdl}. They developed lower bounds on the minimum distances of the duals of narrow-sense cyclic BCH codes of length $\frac{q^m-1}{N}$, where $N$ is a positive integer satisfying $N\mid (q-1)$ or $N=q^s-1$ with $s\mid m$. The question whether the dual code of a cyclic BCH code is BCH is in general very hard to answer. In order to further study the properties of the duals of cyclic BCH codes, the authors in \cite{gdl} proposed the concept of cyclic dually-BCH code: A cyclic BCH code is called a dually-BCH code if both the cyclic BCH code and its dual are BCH codes with respect to an $n$-th primitive root of unity $\beta$. In \cite{gdl,wwlw,wxz}, the authors gave sufficient and necessary conditions for narrow-sense cyclic BCH codes being cyclic dually-BCH codes, where the lengths of the considered codes are $q^m-1$, $\frac{q^m-1}{q-1}$, $\frac{q^m-1}{q+1}$ and $\frac{q^m-1}{N}$, where $N$ is as mentioned above.

	Negacyclic BCH codes are another interesting types of constacyclic BCH codes. In some cases, negacyclic BCH codes have better parameters than cyclic BCH codes \cite{wsd}. A lot of quantum codes with good parameters have been constructed from negacyclic BCH codes \cite{gllw,kz,kzt,zsl}. Zhu et al. \cite{zsl} determined the dimensions of negacyclic BCH codes of length $\frac{q^{2m}-1}{q-1}$ with designed distance $\delta\leq q^m+2$. Pang et al. \cite{pzs} studied three classes of negacyclic BCH codes with designed distance in some ranges and found many optimal codes. Recently, Wang et al. \cite{wsd} analysed the parameters of negacyclic BCH codes of lengths $\frac{q^m-1}{2}$ and $\frac{q^m+1}{2}$ with small and large dimensions, and determined the minimum distances of negacyclic BCH codes of lengths $\frac{q^m-1}{2}$ and $\frac{q^m+1}{2}$ with generator polynomials being an irreducible polynomial or the product of two irreducible polynomials. 
	
	A linear code $\mathcal{C}$ over $\mathbb{F}_q$ is called an LCD (linear complementary dual) code provided $\mathcal{C}\cap \mathcal{C}^{\perp}={\bf 0}$. LCD codes over finite fields have important application in cryptography to resist side-channel attacks and fault non-invasive attacks \cite{bccgh, cg}. \cite{ldl1} and \cite{pzs} proved that cyclic codes of length $q^m+1$ over $\mathbb{F}_q$ and negacyclic codes of length $\frac{q^m+1}{2}$ over $\mathbb{F}_q$ are LCD codes, respectively.

	Inspired and motivated by the above works, in this paper we study the duals of narrow-sense cyclic BCH codes of length $q^m+1$ and the duals of narrow-sense negacyclic BCH codes of length $\frac{q^m+1}{2}$ over $\mathbb{F}_q$. The rest of this paper is organized as follows. Section 2 presents some notions and results that will be used later. Section 3 develops lower bounds on the minimum distances of the dual codes of narrow-sense cyclic BCH codes of length $q^m+1$ over $\mathbb{F}_q$, where $q$ is an odd prime power, and gives sufficient and necessary conditions in terms of the designed distances for the even-like subcodes of narrow-sense cyclic BCH codes of length $q^m+1$ over $\mathbb{F}_q$ being cyclic dually-BCH codes, where $q$ is an odd prime power and $m$ is odd or $m\equiv 2~({\rm mod}~4)$. Section 4 establishes lower bounds on the minimum distances of the dual codes of narrow-sense negacyclic BCH codes of length $\frac{q^m+1}{2}$ over $\mathbb{F}_q$, and present sufficient and necessary conditions in terms of the designed distances to ensure that narrow-sense negacyclic BCH codes of length $\frac{q^m+1}{2}$ over $\mathbb{F}_q$ are negacyclic dually-BCH codes, where $q$ is an odd prime power satisfying $q\equiv 3~({\rm mod}~4)$. Section 5 concludes the paper.	
	

	\section{Preliminaries}
	In this section, we present some basic concepts and results that will be used later. Starting from now on, we adopt the following notions unless otherwise stated:
	
	\begin{itemize}
		\item $\mathbb{F}_q$ is the finite field with $q$ elements.
		
		\item $n$ is a positive integer coprime to $q$. $\mathbb{Z}_n^{*}$ denotes the group of units in $\mathbb{Z}_n$.
		
		\item $r$ is the order of $\lambda$ in $\mathbb{F}_q^{*}$ and $l={\rm ord}_{rn}(q)$ is the order of $q$ in  $\mathbb{Z}_{rn}^{*}$. $\beta$ is a $rn$-th primitive root of unity in $\mathbb{F}_{q^l}$ such that $\beta^n=\lambda$.
		
		\item For integer $i$ with $0\leq i\leq n-1$, the $q$-cyclotomic coset of $1+ri$ modulo $rn$ is defined by
		$$C_{1+ri}^{(q,rn)}=\big\{1+ri, (1+ri)q, (1+ri)q^2, \cdots, (1+ri)q^{k_i-1}\big\}~{\rm mod}~rn\subseteq 1+r\mathbb{Z}_{rn},$$
		where $k_i$ is the least positive integer such that $(1+ri)q^{k_i}\equiv 1+ri~({\rm mod}~rn)$. The smallest integer in $C_{1+ri}^{(q,rn)}$ is called the coset leader of $C_{1+ri}^{(q,rn)}$.
		
		\item For integer $\delta$ with $2\leq \delta\leq n-1$, $\mathcal{C}_{(q,n,\lambda,\delta,1)}$ denotes the narrow-sense $\lambda$-constacyclic BCH code of length $n$ over $\mathbb{F}_q$ with defining set $T=\cup_{i=0}^{\delta-2}C_{1+ri}^{(q,rn)}$ with respect to $\beta$.
		
		\item $T^{\perp}$ is the defining set of the dual code $\mathcal{C}_{(q,n,\lambda,\delta,1)}^{\perp}$ with respect to $\beta$.
		
		\item For integer $a\in 1+r\mathbb{Z}_{rn}$ and integer $t$ with $0\leq t< {\rm ord}_{rn}(q)$, $[aq^t]_{rn}$ denotes the integer in $1+r\mathbb{Z}_{rn}$ which is congruent to $aq^t$ modulo $rn$.
		
		\item $[u,v]:=\{u,u+1,\cdots,v\}$, where $u,v$ are nonnegative integers with $u\leq v$.
		
		\item $\delta_i$ is the $i$-th largest coset leader modulo $q^m+1$, and $\phi_i$ is the $i$-th largest odd coset leader modulo $q^m+1$.
	\end{itemize} 
	
	Constacyclic codes over finite fields have the following BCH bound.
	\begin{Lemma}{\rm \cite{ks}} {\rm \bf(BCH bound for constacyclic codes)}\label{l2.1}
		Let $\mathcal{C}$ be a $\lambda$-constacyclic code of length $n$ over $\mathbb{F}_{q}$ with defining set $T$. If $\{1+ri\mid h\leq i\leq h+\delta-2\}\subseteq T$, where $h\in \mathbb{Z}$, then $\mathcal{C}$ has minimum distance $d\geq \delta$. 
	\end{Lemma}
	
	It is a very hard question when is the dual code of a negacyclic BCH code still a BCH code, although it may be answered for some special negacyclic BCH codes. In order to further investigate the duals of negacyclic BCH codes, we now extend the definition of cyclic dually-BCH codes proposed in \cite{gdl} to negacyclic dually-BCH codes.
	
	\begin{Definition}{\rm \bf(Negacyclic dually-BCH code)}
		If the dual of a negacyclic BCH code $\mathcal{C}$ over $\mathbb{F}_q$ with respect to a $2n$-th primitive root of unity $\beta$ is still a BCH code with respect to $\beta$, then $\mathcal{C}$ is called a dually-BCH code.
	\end{Definition}
	
	The following two lemmas on coset leaders modulo $q^m+1$ will play an important role in the sequel.
		\begin{Lemma}{\rm \cite{ylly,zswh}}\label{l4.0}
		Suppose $q$ is an odd prime power. Let $\delta_1$ and $\delta_2$ be the first two largest coset leaders modulo $q^m+1$, respectively. Then $\delta_{1}=\frac{q^m+1}{2}$ with $C_{\delta_1}=\{\delta_1\}$ and
		
		{\rm (1)} $\delta_2=\frac{(q-1)(q^m+1)}{2(q+1)}$ if $m$ is odd;
		
		{\rm (2)} $\delta_2=\frac{(q-1)^2(q^m+1)}{2(q^2+1)}$ if $m\equiv 2~({\rm mod}~4)$.
	\end{Lemma}
	
	\begin{Lemma}{\rm \cite{wsd}}
		Suppose $m\geq 2$ and $q$ is an odd prime power satisfying $q\equiv 3~({\rm mod}~4)$. Let $\phi_1$, $\phi_2$ and $\phi_3$ be the first three largest odd coset leaders modulo $q^m+1$, respectively. Then
		\begin{align*}
			\phi_1=\begin{cases}
				\frac{q^m+1}{2}, & {\rm if~m~is~even};\\
				\frac{(q-1)(q^m+1)}{2(q+1)}, & {\rm if~m~is~odd},
			\end{cases}
			~~~
			\phi_2=\begin{cases}
				\frac{q^m-1}{2}-q^{m-1}, & {\rm if~m~is~even};\\
				\frac{(q-1)(q^m-2q^{m-2}-1)}{2(q+1)}, & {\rm if~m~is~odd},
			\end{cases}
		\end{align*}
		and if $q^m\geq 25$, then
		\begin{align*}
			\phi_3=\begin{cases}
				\frac{q^m-1}{2}-q^{m-1}-q+1, & {\rm if~m~is~even};\\
				\frac{(q-1)(q^m-2q^{m-2}-1)}{2(q+1)}-(q-1)^2, & {\rm if~m~is~odd}~{\rm and}~m\geq 5;\\
				\frac{(q-1)(q^3-2q-1)}{2(q+1)}-q-1, & {\rm if}~m=3.
			\end{cases}
		\end{align*}
		Moreover, $C_{\phi_1}=\{\phi_1\}$ if $m$ is even and $C_{\phi_1}=\{\phi_1, 2n-\phi_1\}$ if $m$ is odd.
	\end{Lemma}
	
	\section{The duals of cyclic BCH codes of length $q^m+1$}
	
	In this section, we assume that $n=q^m+1$ is the length of the considered codes, where $q\geq 3$ is an odd prime power. For $2 \leq \delta \leq \delta_1$, let $\mathcal{C}_{(q,n,1,\delta,1)}$ be the narrow-sense cyclic BCH code of length $n$ over $\mathbb{F}_q$ with designed distance $\delta$ with respect to a $n$-th primitive root of unity $\beta$. By definition, the defining set of  $\mathcal{C}_{(q,n,1,\delta,1)}$ with respect to $\beta$ is $T=C_1^{(q,n)}\cup C_2^{(q,n)}\cup \cdots \cup C_{\delta-1}^{(q,n)}$. Let $T^{\perp}$ be the defining set of the dual code $\mathcal{C}_{(q,n,1,\delta,1)}^{\perp}$ with respect to $\beta$. It is easy to check that $T^{\perp}=\mathbb{Z}_{n}\setminus T$ and $\delta_1\in T^{\perp}$. In the following, we will develop a lower bound on the minimum distance of the dual code $\mathcal{C}_{(q,n,1,\delta,1)}^{\perp}$ and give a sufficient and necessary condition for the even-like subcode $\mathcal{C}_{(q,n,1,\delta+1,0)}$ of $\mathcal{C}_{(q,n,1,\delta,1)}$ being a cyclic dually-BCH code. To this end, we need the lemma below.
	
	\begin{Lemma}\label{l4.1}
	For $2\leq \delta\leq \delta_1$, let $I(\delta)$ be the integer satisfying $1\leq I(\delta)<\delta_1$ such that $I(\delta) \notin T^{\perp}$ and $\{I(\delta)+1, I(\delta)+2,\cdots, \delta_1\}\subseteq T^{\perp}$. Then
		\begin{small}
		\begin{align*}
			I(\delta)=\begin{cases}
				\frac{q^m-q^{m-\ell_0+1}}{2}+(\ell_1+1)q^{m-\ell_0}, & {\rm if}~\delta=\frac{q^{\ell_0}-q}{2}+\ell_1+2~(1\leq \ell_0\leq m-1,0\leq \ell_1\leq \frac{q-3}{2});\\
				\frac{q^m-q^{m-\ell_0}}{2}+\ell_1+1, & {\rm if}~\ell_1q^{\ell_0}+\frac{q^{\ell_0}+3}{2}\leq \delta <(\ell_1+1)q^{\ell_0}+\frac{q^{\ell_0}+3}{2}\\
				& (1\leq \ell_0\leq m-1,0\leq \ell_1<\frac{q-3}{2});\\
				\frac{q^m-q^{m-\ell_0}+q-1}{2}, & {\rm if}~\frac{q^{\ell_0+1}-2q^{\ell_0}+3}{2}\leq \delta<\frac{q^{\ell_0+1}-q}{2}+2~(1\leq \ell_0\leq m-2);\\
				\frac{q^m-1}{2}, & {\rm if}~\frac{q^m-2q^{m-1}+3}{2}\leq \delta\leq \delta_1.
			\end{cases}
		\end{align*}
		\end{small}
	\end{Lemma}
	
	\begin{proof}
		We prove the desired result only for the case $\delta=\frac{q^{\ell_0}-q}{2}+\ell_1+2$, where $1\leq \ell_0\leq m-1$ and $0\leq \ell_1\leq \frac{q-3}{2}$, as the proof when $\delta$ is in other cases follows the similar way. It is clear that $\frac{q^m-q^{m-\ell_0+1}}{2}+(\ell_1+1)q^{m-\ell_0}=\big(\frac{q^{\ell_0}-q}{2}+\ell_1+1\big)q^{m-\ell_0}\in C_{\frac{q^{\ell_0}-q}{2}+\ell_1+1}^{(q,n)}\nsubseteq T^{\perp}=\mathbb{Z}_n\setminus \big(C_{1}^{(q,n)}\cup C_{2}^{(q,n)}\cup\cdots \cup C_{\delta-1}^{(q,n)}\big)$. 
		
		We next prove that
		\begin{small}
		$$\Gamma:=\Big\{\frac{q^m\!\!-\!q^{m-\ell_0+1}}{2}\!+\!(\ell_1\!+\!1)q^{m-\ell_0}\!+\!1,\frac{q^m\!\!-\!q^{m-\ell_0+1}}{2}\!+\!(\ell_1\!+\!1)q^{m-\ell_0}\!+\!2,\cdots, \delta_1\Big\}\subseteq T^{\perp}.$$
		\end{small}
		For $a\in \Gamma$, $a=\frac{q^m-q^{m-\ell_0+1}}{2}+(\ell_1+1)q^{m-\ell_0}+u$, where $1\leq u\leq \frac{q^{m-\ell_0+1}+1}{2}-(\ell_1+1)q^{m-\ell_0}=(\frac{q-3}{2}-\ell_1)q^{m-\ell_0}+\frac{q-1}{2}(q^{m-\ell_0-1}+q^{m-\ell_0-2}+\cdots+q+1)+1$. Suppose $u=\sum_{j=0}^{m-\ell_0}a_jq^j$, where $0\leq a_{m-\ell_0}\leq \frac{q-3}{2}-\ell_1$ and $0\leq a_j\leq q-1$ for $0\leq j\leq m-\ell_0-1$. To prove $a\in T^{\perp}$, we need to prove that $[aq^t]_n\geq \delta$ for every integer $t$ with $0\leq t\leq 2m-1$. The proof is carried out by distinguishing the following six cases according to the value of $t$.
		
		{\bf Case 1.} $t=0$. It is obvious that $[aq^t]_n=a>\delta$.
		
		{\bf Case 2.} $1\leq t\leq \ell_0-1$. In this case,
		\begin{small}
		\begin{align*}
			aq^t&=\Big(\frac{(q^t-1)q^{m-t}}{2}+\frac{q^{m-t}-q^{m-\ell_0+1}}{2}+(\ell_1+1)q^{m-\ell_0}+u\Big)q^t\\
			&\equiv \frac{q^m-q^{m-\ell_0+1+t}}{2}+(\ell_1+1)q^{m-\ell_0+t}+uq^t-\frac{q^t-1}{2}\\
			&=\frac{q^m-\big(q^{m-\ell_0+1}-(2\ell_1+2)q^{m-\ell_0}-2u+1\big)q^t+1}{2}\triangleq M_t~({\rm mod}~n).
		\end{align*}
		\end{small}
		
		\noindent Direct calculations show that $M_t\leq \frac{q^m+1}{2}<n$, and
		\begin{small}
		$$M_t\geq \frac{q^m-\big(q^{m-\ell_0+1}-(2\ell_1+2)q^{m-\ell_0}-1\big)q^{\ell_0-1}+1}{2}=(\ell_1+1)q^{m-1}+\frac{q^{\ell_0-1}+1}{2}>0.$$
		\end{small}
		Consequently, $[aq^t]_n=M_t>\delta$.
		
		{\bf Case 3.} $\ell_0\leq t\leq m-1$. We have
		\begin{small}
		\begin{align*}
			aq^t&=\Big(\frac{(q^t-q^{t-\ell_0+1})q^{m-t}}{2}+(\ell_1+1)q^{m-\ell_0}+\sum_{j=m-t}^{m-\ell_0}a_jq^j+\sum_{j=0}^{m-t-1}a_jq^j\Big)q^t\\
			&\equiv -\frac{q^t-q^{t-\ell_0+1}}{2}-(\ell_1+1)q^{t-\ell_0}-\sum_{j=0}^{t-\ell_0}a_{j+m-t}q^j+\sum_{j=t}^{m-1}a_{j-t}q^j\triangleq M_t~({\rm mod}~n).
		\end{align*}
		\end{small}
		
		If there is an integer $j$ with $t\leq j\leq m-1$ such that $a_{j-t}\neq 0$, then 
		\begin{small}
		\begin{align*}
			M_t&\geq q^t-\frac{q^t-q^{t-\ell_0+1}}{2}-(\ell_1+1)q^{t-\ell_0}-\big(\frac{q-3}{2}-\ell_1\big)q^{t-\ell_0}-\sum_{j=0}^{t-\ell_0-1}(q-1)q^j\\
			&=\frac{(q^{\ell_0}-1)q^{t-\ell_0}}{2}+1\geq \frac{q^{\ell_0}-1}{2}+1=\frac{q^{\ell_0}+1}{2}>0,
		\end{align*}
		\end{small}
		
		\noindent and
		\begin{small}
		\begin{align*}
			M_t&\leq \sum_{j=t}^{m-1}(q-1)q^j-\frac{q^t-q^{t-\ell_0+1}}{2}-(\ell_1+1)q^{t-\ell_0}=q^m-\big(\frac{3q^{\ell_0}+1}{2}+\ell_1\big)q^{t-\ell_0}\\
			&\leq q^m-\big(\frac{3q^{\ell_0}+1}{2}+\ell_1\big)=q^m-\frac{3q^{\ell_0}+1}{2}-\ell_1<n.
		\end{align*}
		\end{small}
		
		\noindent Therefore, $[aq^t]_n=M_t\geq \delta$.
		
		If $a_{j-t}=0$ for every integer $j$ with $t\leq j\leq m-1$, then $aq^t\equiv q^m+1+M_t~({\rm mod}~n)$. In addition,
		\begin{small}
		\begin{align*}
			q^m+1+M_t&\geq q^m+1-\frac{q^t\!\!-q^{t-\ell_0+1}}{2}-(\ell_1\!+\!1)q^{t-\ell_0}-(\frac{q\!-\!3}{2}-\ell_1)q^{t-\ell_0}-\sum_{j=0}^{t-\ell_0-1}(q\!-\!1)q^j\\
			&=q^m-\frac{q^t+q^{t-\ell_0}}{2}+2\geq q^m-\frac{q^{m-1}+q^{m-\ell_0-1}}{2}+2>0,
		\end{align*}
		\end{small}
		
		\noindent and
		\begin{small}
		\begin{align*}
			q^m+1+M_t&\leq q^m+1-\frac{q^t-q^{t-\ell_0+1}}{2}-(\ell_1+1)q^{t-\ell_0}-1\\
			&=q^m-\big(\frac{q^{\ell_0}-q}{2}+\ell_1+1\big)q^{t-\ell_0}\leq q^m-\frac{q^{\ell_0}-q}{2}-\ell_1-1<n.
		\end{align*}
		\end{small}
		
		\noindent Thus $[aq^t]_n=q^m+1+M_t>\delta$.
		
		{\bf Case 4.} $t=m$. It is easy to see that $[aq^t]_n=n-a\geq q^m+1-\frac{q^m+1}{2}=\frac{q^m+1}{2}>\delta$.
		
		{\bf Case 5.} $m+1\leq t\leq m+\ell_0-1$. Let $t'=t-m$, then $1\leq t'\leq \ell_0-1$ and $[aq^t]_n=n-[aq^{t'}]_n$. We see from Case 2 that $[aq^t]_n\geq q^m+1-\frac{q^m+1}{2}=\frac{q^m+1}{2}>\delta$.
		
		{\bf Case 6.} $m+\ell_0\leq t\leq 2m-1$. Let $t'=t-m$, then $\ell_0\leq t'\leq m-1$ and $[aq^t]_n=n-[aq^{t'}]_n$. It follows from Case 3 that if there is an integer $j$ with $t'\leq j\leq m-1$ such that $a_{j-t'}\neq 0$, then $[aq^t]_n\geq q^m+1-\big(q^m-\frac{3q^{\ell_0}+1}{2}-\ell_1\big)=\frac{3q^{\ell_0}+3}{2}+\ell_1>\delta$, and that if $a_{j-t'}=0$ for every integer $j$ with $t'\leq j\leq m-1$, then $[aq^t]_n\geq q^m+1-\big(q^m-\frac{q^{\ell_0}-q}{2}-\ell_1-1\big)=\frac{q^{\ell_0}-q}{2}+\ell_1+2=\delta$.
		
		Summarizing all the conclusions above, we get the desired result.
	\end{proof}
	
	\begin{Theorem}\label{t4.1}
		Suppose $q\geq 3$ is an odd prime power and $n=q^m+1$. For $2\leq \delta\leq \delta_1$, we have
		\begin{small}
		\begin{align*}
			d(\mathcal{C}_{(q,n,1,\delta,1)}^{\perp})\geq\begin{cases}
				q^{m-\ell_0+1}\!-\!(2\ell_1\!+\!2)q^{m-\ell_0}+1, & {\rm if}~\delta=\frac{q^{\ell_0}-q}{2}+\ell_1+2\\
				&(1\leq \ell_0\leq m-1,~0\leq \ell_1\leq \frac{q-3}{2});\\
				q^{m-\ell_0}-2\ell_1-1, & {\rm if}~\ell_1q^{\ell_0}+\frac{q^{\ell_0}+3}{2}\leq \delta <(\ell_1\!+\!1)q^{\ell_0}+\frac{q^{\ell_0}+3}{2}\\
				& (1\leq \ell_0\leq m-1,~0\leq \ell_1<\frac{q-3}{2});\\
				q^{m-\ell_0}-q+2, & {\rm if}~\frac{q^{\ell_0+1}-2q^{\ell_0}+3}{2}\leq \delta<\frac{q^{\ell_0+1}-q}{2}+2\\
				&(1\leq \ell_0\leq m-2);\\
				2, & {\rm if}~\frac{q^m-2q^{m-1}+3}{2}\leq \delta\leq \delta_1.
			\end{cases}
		\end{align*}
		\end{small}
	\end{Theorem}
	
	\begin{proof}
		For $2\leq \delta\leq \delta_1$, the set $\{I(\delta)+1, I(\delta)+2,\cdots, \delta_1,\delta_1+1,\cdots,n-I(\delta)-1\}$ is a subset of $T^{\perp}$, which contains $n-2I(\delta)-1$ consecutive elements; thus $d(\mathcal{C}_{(q,n,1,\delta,1)}^{\perp})\geq n-2I(\delta)$ according to the BCH bound for cyclic codes. The desired result follows directly from Lemma \ref{l4.1}.
	\end{proof}
	
	\begin{Example}{\rm 
			Take $q=3$ and $m=2$. By Theorem \ref{t4.1}, the lower bounds on the minimum distances of $\mathcal{C}_{(q,n,1,2,1)}^{\perp}$ and $\mathcal{C}_{(q,n,1,3,1)}^{\perp}$ are $4$ and $2$, respectively. By Magma, the true minimum distances of $\mathcal{C}_{(q,n,1,2,1)}^{\perp}$ and $\mathcal{C}_{(q,n,1,3,1)}^{\perp}$ are  $4$ and $2$, respectively.}
	\end{Example}
	
	\begin{Example}{\rm 
			Take $q=5$ and $m=2$. By Theorem \ref{t4.1}, the lower bounds on the minimum distances of $\mathcal{C}_{(q,n,1,2,1)}^{\perp}$ and $\mathcal{C}_{(q,n,1,8,1)}^{\perp}$ are $16$ and $4$, respectively. By Magma, the true minimum distances of $\mathcal{C}_{(q,n,1,2,1)}^{\perp}$ and $\mathcal{C}_{(q,n,1,8,1)}^{\perp}$ are  $16$ and $4$, respectively.}
	\end{Example}
	
		Let $\widetilde{T}$ and $\widetilde{T}^{\perp}$ be the defining sets of the even-like subcode $\mathcal{C}_{(q,n,1,\delta+1,0)}$ and its dual code $\mathcal{C}_{(q,n,1,\delta+1,0)}^{\perp}$ with respect to $\beta$, respectively. By definition, $\widetilde{T}=C_0^{(q,n)}\cup C_1^{(q,n)}\cup C_2^{(q,n)} \cup \cdots \cup C_{\delta-1}^{(q,n)}$, and it is easy to see that $\widetilde{T}^{\perp}=\mathbb{Z}_n\setminus \widetilde{T}$. For $2\leq \delta\leq \delta_1$, $\delta_1\in \widetilde{T}^{\perp}$. Note that $C_{\delta_1}^{(q,n)}=\{\delta_1\}$ and that $C_j^{(q,n)}=C_{n-j}^{(q,n)}$ for $1\leq j<n$, thus $\mathcal{C}_{(q,n,1,\delta+1,0)}$ is a cyclic dually-BCH code if and only if there is an integer $J$ with $1<J<\delta_1$ such that $\widetilde{T}^{\perp}=C_J^{(q,n)}\cup C_{J+1}^{(q,n)}\cup \cdots\cup C_{\delta_1}^{(q,n)}$. We next present a sufficient and necessary condition in terms of $\delta$ for $\mathcal{C}_{(q,n,1,\delta+1,0)}$ being a cyclic dually-BCH code.
		
	\begin{Lemma}\label{l4.2}
		Suppose $q>3$ is an odd prime power and $m=2$. For $3\leq \delta\leq \delta_2-1$, let $M(\delta)=\widetilde{T}^{\perp}\setminus (C_{I(\delta)+1}^{(q,n)}\cup C_{I(\delta)+2}^{(q,n)}\cup \cdots \cup C_{\delta_1}^{(q,n)})$. Then $q+2\in M(\delta)$ if $3\leq \delta\leq q+2$ and $\delta_2-1\in M(\delta)$ if $q+2<\delta\leq \delta_2-1$.
	\end{Lemma}
	\begin{proof}
		When $3\leq \delta\leq q+2$, it is obvious that $q+2\in \widetilde{T}^{\perp}$. We see from Lemma \ref{l4.1} that ${\rm min}\{I(\delta)\mid 3\leq \delta\leq q+2\}=2q$. As $C_{q+2}^{(q,n)}=\{q+2,2q-1,q^2-q-1,q^2-2q+2\}$, we have $C_{q+2}^{(q,n)}\cap \big(C_{2q+1}^{(q,n)}\cup C_{2q+2}^{(q,n)}\cup \cdots \cup C_{\delta_1}^{(q,n)}\big)=\emptyset$, which implies that $q+2\in M(\delta)$.
		
		When $q+2<\delta\leq \delta_2-1$, we first have $\delta_2-1\in \widetilde{T}^{\perp}$. It follows from Lemma \ref{l4.1} that ${\rm min}\{I(\delta)\mid q+2<\delta\leq \delta_2-1\}=\frac{q^2-q}{2}+1$. Note that $C_{\delta_2-1}^{(q,n)}=C_{\frac{q^2-2q-1}{2}}^{(q,n)}=\big\{\frac{q^2-2q-1}{2}, \frac{q^2-2q+3}{2},\frac{q^2+2q+3}{2},\frac{q^2+2q-1}{2}\big\}$. One can check that $C_{\delta_2-1}^{(q,n)}\cap \big(C_{\frac{q^2-q}{2}+2}^{(q,n)}\cup C_{\frac{q^2-q}{2}+3}^{(q,n)}\cup \cdots \cup C_{\delta_1}^{(q,n)}\big)=\emptyset$, and hence $\delta_2-1\in M(\delta)$.
	\end{proof}
	
	\begin{Theorem}\label{t4.2}
		Suppose $q\geq 3$ is an odd prime power and $m=2$. For $2\leq \delta\leq \delta_1$, $\mathcal{C}_{(q,n,1,\delta+1,0)}$ is a cyclic dually-BCH code if and only if
		$$\delta=2~{\rm or}~\delta_2\leq \delta\leq \delta_1.$$
	\end{Theorem}
	
	\begin{proof}
		When $\delta=2$, $\widetilde{T}^{\perp}=\mathbb{Z}_n\setminus \big(C_0^{(q,n)}\cup C_1^{(q,n)}\big)$. It follows from Lemma \ref{l4.1} that $I(\delta)=q$. Thus $C_{q+1}^{(q,n)}\cup C_{q+2}^{(q,n)}\cup \cdots\cup C_{\delta_1}^{(q,n)}\subseteq \widetilde{T}^{\perp}$. We claim that $\widetilde{T}^{\perp}=C_{q+1}^{(q,n)}\cup C_{q+2}^{(q,n)}\cup \cdots\cup C_{\delta_1}^{(q,n)}(=C_{q+1}^{(q,n)}\cup C_{q+2}^{(q,n)}\cup \cdots\cup C_{\delta_1}^{(q,n)}\cup C_{\delta_1+1}^{(q,n)}\cup \cdots\cup C_{q^2-q}^{(q,n)})$, and hence $\mathcal{C}_{(q,n,1,\delta+1,0)}^{\perp}=\mathcal{C}_{(q,n,1,q^2-2q+1,q+1)}$ is a BCH code with designed distance $q^2-2q+1$ with respect to $\beta$. To prove this, we need to prove that if $a\notin C_0^{(q,n)}\cup C_1^{(q,n)}=\{0, 1,q,q^2,q^2-q+1\}$, then $a\in C_{q+1}^{(q,n)}\cup C_{q+2}^{(q,n)}\cup \cdots\cup C_{\delta_1}^{(q,n)}$. Since $C_j=C_{n-j}$ for $1\leq j<n$, we only need consider integer $a$ satisfying $2\leq a\leq q-1$ or $q+1\leq a\leq \frac{q^2+1}{2}$. If $a$ is in the latter case, the result is obviously true. If $2\leq a\leq q-1$, then $q+1\leq n-aq\leq q^2-2q+1$; the result also holds as $a\in C_{n-aq}$.
		
		When $3\leq \delta\leq \delta_2-1$, we know from Lemma \ref{l4.2} that $C_{I(\delta)+1}^{(q,n)}\cup C_{I(\delta)+2}^{(q,n)}\cup \cdots \cup C_{\delta_1}^{(q,n)}\subsetneq \widetilde{T}^{\perp}$, which implies that there is no integer $J$ with $1<J<\delta_1$ such that $\widetilde{T}^{\perp}=C_J^{(q,n)}\cup C_{J+1}^{(q,n)}\cup \cdots\cup C_{\delta_1}^{(q,n)}$, i.e., $\mathcal{C}_{(q,n,1,\delta+1,0)}$ is not a dually-BCH code.
		
		When $\delta=\delta_2$, $\widetilde{T}^{\perp}=C_{\delta_2}^{(q,n)}\cup C_{\delta_1}^{(q,n)}=C_{\frac{(q-1)^2}{2}}^{(q,n)}\cup C_{\frac{q^2+1}{2}}^{(q,n)}=C_{\frac{q^2-1}{2}}^{(q,n)}\cup C_{\frac{q^2+1}{2}}^{(q,n)}\cup C_{\frac{q^2+3}{2}}^{(q,n)}$, and thus $\mathcal{C}_{(q,n,1,\delta+1,0)}^{\perp}=\mathcal{C}_{(q,n,1,4,\frac{q^2-1}{2})}$ is a BCH code with designed distance $4$ with respect to $\beta$.
		
		When $\delta_2< \delta\leq \delta_1$, $\widetilde{T}^{\perp}=C_{\delta_1}^{(q,n)}$ and $\mathcal{C}_{(q,n,1,\delta+1,0)}^{\perp}=\mathcal{C}_{(q,n,1,2,\delta_1)}$ is a BCH code with designed distance $2$ with respect to $\beta$. 
	\end{proof}
	
		\begin{Example}{\rm 
			Take $q=5$ and $m=2$. By Magma, for $2\leq \delta\leq \delta_1=13$,  $\mathcal{C}_{(q,n,1,\delta+1,0)}$ is a cyclic dually-BCH code if and only if $\delta=2$ or $8 \leq \delta \leq 13$, which coincides with Theorem \ref{t4.2}.}
	\end{Example}
	
	\begin{Lemma}\label{l4.3}
		Suppose $q\geq 3$ is an odd prime power and $m$ is odd, or $m\geq 6$ and $m\equiv 2~({\rm mod}~4)$. For $2\leq \delta\leq \delta_2$, let $M(\delta)=\widetilde{T}^{\perp}\setminus (C_{I(\delta)+1}^{(q,n)}\cup C_{I(\delta)+2}^{(q,n)}\cup \cdots \cup C_{\delta_1}^{(q,n)})$. Then $q^2-q+1\in M(\delta)$ if $2\leq \delta\leq q^2-q+1$ and $\delta_2\in M(\delta)$ if $q^2-q+1<\delta\leq \delta_2$.
	\end{Lemma}
	\begin{proof}
		When $2\leq \delta\leq q^2-q+1$, $q^2-q+1\in \widetilde{T}^{\perp}$. In addition, we see from Lemma \ref{l4.1} that ${\rm min}\{I(\delta)\mid 2\leq \delta\leq q^2-q+1\}=q^{m-1}$. To prove $q^2-q+1\in M(\delta)$, we need to prove that either $[(q^2-q+1)q^t]_n\leq q^{m-1}$ or $[(q^2-q+1)q^t]_n>\delta_1=\frac{q^m+1}{2}$ for every integer $t$ with $0\leq t\leq 2m-1$. The proof is divided into the following six cases.
		
		{\bf Case 1.} $0\leq t\leq m-3$. $[(q^2-q+1)q^t]_n=(q^2-q+1)q^t\leq (q^2-q+1)q^{m-3}=q^{m-1}-q^{m-2}+q^{m-3}<q^{m-1}$.
		
		{\bf Case 2.} $t=m-2$. $[(q^2-q+1)q^t]_n=(q^2-q+1)q^t=q^m-q^{m-1}+q^{m-2}>\delta_1$.
		
		{\bf Case 3.} $t=m-1$. In this case, $(q^2-q+1)q^t=(q-1)q^m+q^{m-1}\equiv q^{m-1}-q+1~({\rm mod}~n)$. Then $[(q^2-q+1)q^t]_n=q^{m-1}-q+1<q^{m-1}$.
		
		{\bf Case 4.} $m\leq t\leq 2m-3$. Let $t'=t-m$, then $0\leq t'\leq m-3$ and $[(q^2-q+1)q^t]_n=n-[(q^2-q+1)q^{t'}]_n$. It follows from Case 1 that $[(q^2-q+1)q^t]_n\geq q^m-q^{m-1}+q^{m-2}-q^{m-3}+1>\delta_1$.
		
		{\bf Case 5.} $t=2m-2$. We see from Case 2 that $[(q^2-q+1)q^t]_n=n-[(q^2-q+1)q^{m-2}]_n=q^{m-1}-q^{m-2}+1<q^{m-1}$.
		
		{\bf Case 6.} $t=2m-1$. From Case 3 we have $[(q^2-q+1)q^t]_n=n-[(q^2-q+1)q^{m-1}]_n=q^{m}-q^{m-1}+q>\delta_1$.
		
		When $q^2-q+1<\delta\leq \delta_2$, it is obvious that $\delta_2\in \widetilde{T}^{\perp}$. It follows from Lemma \ref{l4.1} that ${\rm min}\{I(\delta)\mid q^2-q+1<\delta\leq \delta_2\}=\frac{q^m-q^{m-2}}{2}+1$. If $m$ is odd, then  $\delta_2=\frac{(q-1)(q^m+1)}{2(q+1)}$ according to Lemma \ref{l4.0}. Note that $C_{\delta_2}^{(q,n)}=\big\{\frac{(q-1)(q^m+1)}{2(q+1)}, \frac{(q+3)(q^m+1)}{2(q+1)}\big\}$. Then it is easy to see that $\delta_2\in M(\delta)$. If $m\geq 6$ and $m\equiv 2~({\rm mod}~4)$, Lemma \ref{l4.0} tells us that $\delta_2=\frac{(q-1)^2(q^m+1)}{2(q^2+1)}$. In this case, one can check that $C_{\delta_2}^{(q,n)}=\big\{\frac{(q-1)^2(q^m+1)}{2(q^2+1)}, \frac{(q^2+3)(q^m+1)}{2(q^2+1)}, \frac{(q^2-1)(q^m+1)}{2(q^2+1)},\frac{(q+1)^2(q^m+1)}{2(q^2+1)}\big\}\subseteq M(\delta)$. This completes the proof.
	\end{proof}

	\begin{Theorem}\label{t4.3}
		Suppose $q\geq 3$ is an odd prime power and $m$ is odd, or $m\geq 6$ and $m\equiv 2~({\rm mod}~4)$. For $2\leq \delta\leq \delta_1$, $\mathcal{C}_{(q,n,1,\delta+1,0)}$ is a cyclic dually-BCH code if and only if
		$$\delta_2+1\leq \delta\leq \delta_1.$$
	\end{Theorem}
	\begin{proof}
		The desired result follows directly from Lemma \ref{l4.3}.
	\end{proof}
	
		\begin{Example}{\rm 
			Take $q=3$ and $m=3$. By Magma, for $2\leq \delta\leq \delta_1=14$,  $\mathcal{C}_{(q,n,1,\delta+1,0)}$ is a cyclic dually-BCH code if and only if $8 \leq \delta \leq 14$, which coincides with Theorem \ref{t4.3}.}
	\end{Example}
	
	\section{The duals of negacyclic BCH codes of length $\frac{q^m+1}{2}$}	
	
	Throughout this section, we always assume that $n=\frac{q^m+1}{2}$ is the length of the considered negacyclic codes, where $m\geq 2$ and $q$ is an odd integer satisfying $q\equiv 3~({\rm mod}~4)$. For $2\leq \delta<\frac{\phi_1+3}{2}$, let $\mathcal{C}_{(q,n,-1,\delta,1)}$ be the narrow-sense negacyclic BCH code of length $n$ over $\mathbb{F}_q$ with designed distance $\delta$ with respect to a $2n$-th primitive root of unity $\beta$. By definition, the defining set of  $\mathcal{C}_{(q,n,-1,\delta,1)}$ with respect to $\beta$ is $T=C_1^{(q,2n)}\cup C_3^{(q,2n)}\cup \cdots \cup C_{1+2(\delta-2)}^{(q,2n)}$. Let $T^{\perp}$ be the defining set of the dual code $\mathcal{C}_{(q,n,-1,\delta,1)}^{\perp}$ with respect to $\beta$. It is easy to check that $T^{\perp}=(1+2\mathbb{Z}_{2n})\setminus T$ and $\phi_1\in T^{\perp}$. In the following, we will give a lower bound on the minimum distance of the dual code $\mathcal{C}_{(q,n,-1,\delta,1)}^{\perp}$ and present a sufficient and necessary condition in terms of $\delta$ for $\mathcal{C}_{(q,n,-1,\delta,1)}$ being a negacyclic dually-BCH code. We treat the cases $q=3$ and $q>3$ seperately as the results of these two cases are distinct.
	
	\subsection*{A. The case $q=3$}
	In this subsection, we consider the duals of ternay narrow-sense negacyclic BCH codes of length $n=\frac{3^m+1}{2}$. Then two subcases arise: $m$ is odd and $m$ is even.
	\subsection*{A.1. The subcase $m$ is odd}
	
	Assume that $q=3$ and $m$ is odd.	To give a lower bound on the minimum distance of $\mathcal{C}_{(q,n,-1,\delta,1)}^{\perp}$ and present a characterization of $\mathcal{C}_{(q,n,-1,\delta,1)}$ being a negacyclic dually-BCH code, we will need the following two lemmas.
	\begin{Lemma}\label{l3.1}
		Suppose $q=3$ and $m\geq 3$ is odd. For $2\leq \delta< \frac{\phi_2+3}{2}$, let $I_1(\delta)$ be the odd integer such that $I_1(\delta)\not\in T^{\perp}$ and $\{I_1(\delta)+2,I_1(\delta)+4,\cdots,\phi_1\}\in T^{\perp}$. Then
		\begin{align*}
			I_1(\delta)=\begin{cases}
				\frac{3^m-5\cdot 3^{m-2\ell}}{4}, & {\rm if}~\frac{3^{2\ell}+7}{8}\leq \delta \leq \frac{3^{2\ell+1}-3}{8};\\
				\frac{3^m-7\cdot 3^{m-2\ell-1}}{4}, & {\rm if}~\delta = \frac{3^{2\ell+1}+5}{8};\\
				\frac{3^m-3^{m-2\ell}+4}{4}, & {\rm if}~\frac{3^{2\ell+1}+13}{8}\leq \delta \leq \frac{3^{2\ell+2}-1}{8};\\
				\frac{3^m-15}{4}, & {\rm if}~\frac{3^{m-1}+7}{8}\leq \delta \leq \frac{19\cdot 3^{m-3}+5}{8};\\
				\frac{3^m-7}{4}, & {\rm if}~\frac{19\cdot 3^{m-3}+13}{8}\leq \delta < \frac{\phi_2+3}{2},
			\end{cases}
			~~{\rm where}~1\leq \ell \leq \frac{m-3}{2}.
		\end{align*}
	\end{Lemma}
	
	\begin{proof}
		We prove the desired conclusions only for the case $\frac{3^{2\ell}+7}{8}\leq \delta \leq \frac{3^{2\ell+1}-3}{8}$, as the other cases can be similarly proved. In this case, $\frac{3^{2\ell}-5}{4}\leq 1+2(\delta-2)\leq \frac{3^{2\ell+1}-15}{4}$. It is clear that $\frac{3^m-5\cdot 3^{m-2\ell}}{4}=\frac{3^{2\ell}-5}{4}\cdot 3^{m-2\ell}\in C_{\frac{3^{2\ell}-5}{4}}^{(q,2n)}\nsubseteq T^{\perp}=(1+2\mathbb{Z}_{2n})\setminus \big(C_1^{(q,2n)}\cup C_3^{(q,2n)}\cup \cdots \cup C_{1+2(\delta-2)}^{(q,2n)}\big)$. 
		
		Next we need to show that $\Gamma:=\big\{\frac{3^m-5\cdot 3^{m-2\ell}}{4}+2, \frac{3^m-5\cdot 3^{m-2\ell}}{4}+4,\cdots, \phi_1\big\}\subseteq T^{\perp}$. For every integer $a\in \Gamma$, we have $a=\frac{3^m-5\cdot 3^{m-2\ell}}{4}+u$, where $u$ is an even integer satisfying $2\leq u\leq \phi_1-\frac{3^m-5\cdot 3^{m-2\ell}}{4}=3^{m-2\ell}+\frac{3^{m-2\ell}+1}{4}$. Suppose $u=\sum_{j=0}^{m-2\ell}a_j\cdot 3^j$, where $a_{m-2\ell}\in \{0,1\}$ and $a_j\in \{0,1,2\}$ with $0\leq j\leq m-2\ell-1$. To prove $a\in T^{\perp}$, it suffices to prove that $[a\cdot 3^t]_{2n}>\frac{3^{2\ell+1}-15}{4}$ for every integer $t$ with $0\leq t\leq 2m-1$. We organize the proof into the following six cases according to the value of $t$.
		
		{\bf Case 1.} $t=0$. It is clear that $[a\cdot 3^t]_{2n}=a\geq \frac{3^m-5\cdot 3^{m-2\ell}}{4}+2>\frac{3^{2\ell+1}-15}{4}$.
		
		{\bf Case 2.} $1\leq t\leq 2\ell-1$. If $t$ is odd, then
		\begin{small}
			$$a\cdot3^t=\frac{3^{m+t}\!-\!5\!\cdot\! 3^{m-2\ell+t}}{4}+\!u\cdot 3^t\equiv \frac{3^{m+1}\!-\!5\!\cdot\! 3^{m-2\ell+t}\!+\!3^t\!-\!3}{4}+\!u\cdot 3^t\triangleq M_t~({\rm mod}~2n).$$
		\end{small}
		Note that $0<M_t\leq \frac{3^{m+1}-5\cdot 3^{m-2\ell+t}+3^t-3}{4}+\frac{5\cdot 3^{m-2\ell}+1}{4}\cdot 3^t=\frac{3^{m+1}+2\cdot 3^t-3}{4}<2n$. Thus
		\begin{small}
			\begin{align*}
				[a\cdot 3^t]_{2n}&=M_t\geq \frac{3^{m+1}-5\cdot 3^{m-2\ell+t}+3^t-3}{4}+2\cdot 3^t=\frac{3^{m+1}-3^t(5\cdot 3^{m-2\ell}-3^2)-3}{4}\\
				&\geq \frac{3^{m+1}-3^{2\ell-1}(5\cdot 3^{m-2\ell}-3^2)-3}{4}=\frac{3^{m}+3^{m-1}+3^{2\ell+1}-3}{4}>\frac{3^{2\ell+1}-15}{4}.
			\end{align*}
		\end{small}
		
		If $t$ is even, then 
		\begin{small}
			$$a\cdot 3^t=\frac{3^{m+t}\!-\!5\!\cdot\! 3^{m-2\ell+t}}{4}+u\cdot 3^t\equiv \frac{3^m\!-\!5\!\cdot\! 3^{m-2\ell+t}\!+\!3^t\!-\!1}{4}+u\cdot 3^t\triangleq M_t~({\rm mod}~2n).$$
		\end{small}
		As $0<M_t\leq \frac{3^m-5\cdot 3^{m-2\ell+t}+3^t-1}{4}+\frac{5\cdot 3^{m-2\ell}+1}{4}\cdot 3^t=\frac{3^m+2\cdot 3^t-1}{4}<2n$, we have
		\begin{small}
			\begin{align*}
				[a\cdot 3^t]_{2n}&=M_t\geq \frac{3^m-5\cdot 3^{m-2\ell+t}+3^t-1}{4}+2\cdot 3^t=\frac{3^m-3^t(5\cdot 3^{m-2\ell}-3^2)-1}{4}\\
				&\geq \frac{3^m-3^{2\ell-2}(5\cdot 3^{m-2\ell}-3^2)-1}{4}=\frac{3^{m-1}+3^{m-2}+3^{2\ell}-1}{4}>\frac{3^{2\ell+1}-15}{4}.
			\end{align*}
		\end{small}
		
		{\bf Case 3.} $2\ell\leq t\leq m-1$. In this case, 
		\begin{small}
			$$a\!\cdot\! 3^t\!=\!\frac{3^{m-2\ell+t}(3^{2\ell}\!-\!5)}{4}+\!\!\sum_{j=t}^{m-2\ell+t}\!\!a_{j-t}3^j\equiv\! \sum_{j=t}^{m-1}\!a_{j-t}3^j\!-\!\frac{3^t\!-\!5\!\cdot\! 3^{t-2\ell}}{4}\!-\!\sum_{j=0}^{t-2\ell}\!a_{j+m-t}3^j\triangleq M_t~({\rm mod}~2n).$$
		\end{small}
		If there is an integer $j$ with $t\leq j\leq m-1$ such that $a_{j-t}\neq 0$, then
		\begin{small}
			$$M_t\geq 3^t-\frac{3^t-5\cdot 3^{t-2\ell}}{4}-\sum_{j=0}^{t-2\ell}2\cdot 3^j=\frac{3^{t-2\ell}(3^{2\ell+1}-7)+4}{4}\geq \frac{3^{2\ell+1}-3}{4}>0,~{\rm and}$$
			$$M_t\leq \sum_{j=t}^{m-1}2\cdot 3^j-\frac{3^t-5\cdot 3^{t-2\ell}}{4}=3^m-3^t-\frac{3^t-5\cdot 3^{t-2\ell}}{4}<2n.$$
		\end{small}
		Thus $[a\cdot 3^t]_{2n}=M_t>\frac{3^{2\ell+1}-15}{4}$. If $a_{j-t}=0$ for every integer $j$ with $t\leq j\leq m-1$, then
		\begin{small}
			\begin{align*}
				[a\cdot 3^t]_{2n}&=2n-\frac{3^t-5\cdot 3^{t-2\ell}}{4}-\sum_{j=0}^{t-2\ell}a_{j+m-t}3^j\geq 3^m+1-\frac{3^t-5\cdot 3^{t-2\ell}}{4}-\sum_{j=0}^{t-2\ell}2\cdot 3^j\\
				&=3^m+2-\frac{3^{t-2\ell}(3^{2\ell}+17)}{4}\geq 3^m+2-\frac{3^{m-1-2\ell}(3^{2\ell}+17)}{4}\\
				&=\frac{3^{m+1}+2\cdot 3^{m-1}-2\cdot 3^{m-2\ell+1}+3^{m-2\ell-1}+8}{4}>\frac{3^{2\ell+1}-15}{4}.
			\end{align*}
		\end{small}
		
		{\bf Case 4.} $t=m$. Then $[a\cdot 3^t]_{2n}=2n-a\geq 3^m+1-\frac{3^m+1}{4}=\frac{3^{m+1}+3}{4}>\frac{3^{2\ell+1}-15}{4}$.
		
		{\bf Case 5.} $m+1\leq t\leq m+2\ell-1$. Let $t'=t-m$, then $1\leq t'\leq 2\ell-1$ and $[a\cdot 3^t]_{2n}=2n-[a\cdot 3^{t'}]_{2n}$. If $t$ is even, then $t'$ is odd and we see from Case 2 that
		\begin{small}
			$$[a\cdot 3^t]_{2n}\geq 2n-\frac{3^{m+1}+2\cdot 3^{t'}-3}{4}=\frac{3^m-2\cdot 3^{t'}+7}{4}\geq \frac{3^m-2\cdot 3^{2\ell-1}+7}{4}>\frac{3^{2\ell+1}-15}{4}.$$
		\end{small}
		If $t$ is odd, then $t'$ is even and from Case 2 we have
		\begin{small}
			$$[a\cdot 3^t]_{2n}\geq 2n-\frac{3^m+2\cdot 3^{t'}-1}{4}=\frac{3^{m+1}-2\cdot 3^{t'}+5}{4}\geq \frac{3^{m+1}-2\cdot 3^{2\ell-2}+5}{4}>\frac{3^{2\ell+1}-15}{4}.$$
		\end{small}
		
		{\bf Case 6.} $m+2\ell\leq t\leq 2m-1$. Let $t'=t-m$, then $2\ell\leq t'\leq m-1$ and $[a\cdot 3^t]_{2n}=2n-[a\cdot 3^{t'}]_{2n}$. According to Case 3, we know that if there is an integer $j$ with $t'\leq j\leq m-1$ such that $a_{j-t'}\neq 0$, then
		\begin{small}
			\begin{align*}
				[a\cdot 3^t]_{2n}&\geq 2n\!-3^m+3^{t'}+\frac{3^{t'}\!\!-\!5\cdot 3^{t'-2\ell}}{4}=\frac{5\cdot 3^{t'-2\ell}(3^{2\ell}\!-\!1)+4}{4}\geq \frac{5(3^{2\ell}\!-\!1)+4}{4}>\frac{3^{2\ell+1}\!-\!15}{4}.
			\end{align*}
		\end{small}
		
		\noindent If $a_{j-t'}=0$ for every integer $j$ with $t'\leq j\leq m-1$, according to Case 3 again we have
		\begin{small}
			$$[a\cdot 3^t]_{2n}=\frac{3^{t'}-5\cdot 3^{t'-2\ell}}{4}+\sum_{j=0}^{t'-2\ell}a_{j+m-t'}3^j.$$
		\end{small}
		In this subcase, we assert that $t'\geq 2\ell+1$; otherwise, $u=3^{m-2\ell}$, which is impossible as $u$ is an even integer. Moreover, when $t'=2\ell+1$, we have $u=a_{m-2\ell}3^{m-2\ell}+a_{m-2\ell-1}3^{m-2\ell-1}$, which forces $(a_{m-2\ell},a_{m-2\ell-1})=(0,2)$ given that $u$ is even and $2\leq u\leq 3^{m-2\ell}+\frac{3^{m-2\ell}+1}{4}$. It follows that
		\begin{small}
			$$[a\cdot 3^t]_{2n}\geq \frac{3^{2\ell+1}-5\cdot 3}{4}+a_{m-2\ell}\cdot 3+a_{m-2\ell-1}\geq \frac{3^{2\ell+1}-5\cdot 3}{4}+2=\frac{3^{2\ell+1}-7}{4}>\frac{3^{2\ell+1}-15}{4}.$$
		\end{small}
		
		The proof is then completed.
	\end{proof}
	
	\begin{Lemma}\label{l3.2}
		Suppose $q=3$ and $m\geq 3$ is odd. For $2\leq \delta< \frac{\phi_1+3}{2}$, let $I_2(\delta)$ be the odd integer such that $\{\phi_1,\phi_1+2,\cdots,I_2(\delta)-2\}\in T^{\perp}$ and $I_2(\delta)\not\in T^{\perp}$. Then
		\begin{small}
			$$I_{2}(\delta)=\frac{3^m+3^{m-2\ell+1}}{4}~{\rm if}~\frac{3^{2\ell-1}+13}{8}\leq \delta \leq \frac{3^{2\ell+1}+5}{8},~{\rm where}~1\leq \ell \leq \frac{m-1}{2}.$$
		\end{small}
	\end{Lemma}
	
	\begin{proof}
		When $\frac{3^{2\ell-1}+13}{8}\leq \delta \leq \frac{3^{2\ell+1}+5}{8}$, $\frac{3^{2\ell-1}+1}{4}\leq 1+2(\delta-2)\leq \frac{3^{2\ell+1}-7}{4}$. It is easy to see that $\frac{3^m+3^{m-2\ell+1}}{4}=\frac{3^{2\ell-1}+1}{4}\cdot 3^{m-2\ell+1}\in C_{\frac{3^{2\ell-1}+1}{4}}^{(q,2n)}\nsubseteq T^{\perp}=(1+2\mathbb{Z}_{2n})\setminus (C_1^{(q,2n)}\cup C_3^{(q,2n)}\cup\cdots\cup C_{1+2(\delta-2)}^{(q,2n)})$. We next prove that
		\begin{small}
			$$\Gamma:=\Big\{\frac{3^m+1}{4},\frac{3^m+1}{4}+2,\cdots,\frac{3^m+3^{m-2\ell+1}}{4}-2\Big\}\subseteq T^{\perp}.$$
		\end{small}
		It is obvious that $\phi_1=\frac{3^m+1}{4}\in T^{\perp}$. For $a\in\Gamma\setminus \{\phi_1\}$, $a=\frac{3^m+1}{4}+u$, where $u$ is an even integer satisfying $2\leq u\leq \frac{3^{m-2\ell+1}-9}{4}=2\cdot 3^{m-2\ell-1}+2\cdot 3^{m-2\ell-3}+\cdots+2\cdot 3^{4}+2\cdot 3^{2}$. Suppose $u=\sum_{j=0}^{m-2\ell-1}a_j\cdot 3^j$, where $a_{j}\in\{0,1,2\}$. We need to prove that $[a\cdot 3^t]_{2n}>\frac{3^{2\ell+1}-7}{4}$ for every integer $t$ with $0\leq t\leq 2m-1$.
		The proof is divided into eight cases.
		
		{\bf Case 1.} $t=0$. Then $[a\cdot 3^t]_{2n}=a>\frac{3^m+1}{4}> \frac{3^{2\ell+1}-7}{4}$.
		
		{\bf Case 2.} $1\leq t\leq 2\ell$. Suppose $t$ is odd, then
		\begin{small}
			$$a\cdot 3^t=\frac{3^m+1}{4}\cdot 3^t+u\cdot 3^t\equiv \frac{3^{m+1}+3}{4}+u\cdot 3^t \triangleq M_t~({\rm mod}~2n).$$
		\end{small}
		As $0<M_t\leq \frac{3^{m+1}+3}{4}+\frac{3^{m-2\ell+1}-9}{4}\cdot 3^t=3^m-\frac{3^{2\ell+1}-3}{4}<2n$, we have $[a\cdot 3^{t}]_{2n}=M_t> \frac{3^{m+1}+3}{4}>\frac{3^{2\ell+1}-7}{4}$.
		
		Suppose $t$ is even, then
		\begin{small}
			$$a\cdot 3^t=\frac{3^m+1}{4}\cdot 3^t+u\cdot 3^t\equiv \frac{3^{m}+1}{4}+u\cdot 3^t \triangleq M_t~({\rm mod}~2n).$$
		\end{small}
		As $0<M_t\leq \frac{3^{m}+1}{4}+\frac{3^{m-2\ell+1}-9}{4}\cdot 3^t=3^m-\frac{3^{2\ell+2}-1}{4}<2n$, $[a\cdot 3^{t}]_{2n}=M_t> \frac{3^{m}+1}{4}>\frac{3^{2\ell+1}-7}{4}$.

		{\bf Case 3.} $2\ell+1\leq t\leq m-2$. Suppose $t$ is odd, then
		\begin{small}
			$$a\cdot 3^t=\frac{3^m\!+\!1}{4}\cdot 3^t+\!\sum_{j=t}^{m-2\ell-1+t}\!a_{j-t}3^{j}\equiv \frac{3^{m+1}\!+\!3}{4}+\sum_{j=t}^{m-1}a_{j-t}3^{j}-\!\sum_{j=0}^{t-2\ell-1}\!a_{j+m-t}3^j\triangleq M_t~({\rm mod}~2n).$$
		\end{small}
		Direct calculations show that 
		\begin{small}
			$$M_t\geq \frac{3^{m+1}\!+\!3}{4}-\!\sum_{j=0}^{t-2\ell-1}\!2\cdot 3^j=\frac{3^{m+1}\!-\!3^{t-2\ell+1}\!-\!3^{t-2\ell}\!+\!7}{4}>0,~{\rm and}$$
		\end{small}
		\begin{small}
			$$M_t\leq \frac{3^{m+1}\!+\!3}{4}+\sum_{j=t}^{m-1}2\cdot 3^j=\frac{2\!\cdot\! 3^{m+1}\!+\!3^m\!-\!3^{t+1}\!-\!3^{t}\!+\!3}{4}\leq \frac{2\!\cdot\! 3^{m+1}\!+\!3^m\!-\!3^{2\ell+2}\!-\!3^{2\ell+1}\!+\!3}{4}<4n.$$
		\end{small}
		If $0<M_t<2n$, then $[a\cdot 3^{t}]_{2n}=M_t\geq \frac{3^{m+1}-3^{m-2\ell-1}-3^{m-2\ell-2}+7}{4}>\frac{3^{2\ell+1}-7}{4}$. If $2n<M_t<4n$, then
		\begin{small}
			\begin{align*}
				[a\cdot 3^{t}]_{2n}&=M_t-2n=\sum_{j=t}^{m-1}a_{j-t}3^{j}-\frac{3^m+1}{4}-\sum_{j=0}^{t-2\ell-1}a_{j+m-t}3^j\\
				&=\sum_{j=t}^{m-1}\big(a_{j-t}+(-1)^{j+1}\big)3^j-\frac{3^t+1}{4}-\sum_{j=0}^{t-2\ell-1}a_{j+m-t}3^j\geq 3^t-\frac{3^t+1}{4}-\sum_{j=0}^{t-2\ell-1}2\cdot 3^j\\
				&=\frac{(3^{2\ell+1}-4)3^{t-2\ell}+3}{4}\geq \frac{(3^{2\ell+1}-4)\cdot 3+3}{4}=\frac{3^{2\ell+2}-9}{4}>\frac{3^{2\ell+1}-7}{4}.
			\end{align*}
		\end{small}
		
		Suppose $t$ is even, then
		\begin{small}
			$$a\cdot 3^t=\frac{3^m+1}{4}\cdot 3^t+\sum_{j=t}^{m-2\ell-1+t}a_{j-t}3^{j}\equiv \frac{3^{m}+1}{4}+\sum_{j=t}^{m-1}a_{j-t}3^{j}-\sum_{j=0}^{t-2\ell-1}a_{j+m-t}3^j\triangleq M_t~({\rm mod}~2n).$$
		\end{small}
		Note that
		\begin{small}
			$$M_t\geq \frac{3^{m}+1}{4}-\sum_{j=0}^{t-2\ell-1}2\cdot 3^j=\frac{3^m-3^{t-2\ell+1}-3^{t-2\ell}+5}{4}>0,~{\rm and}$$
		\end{small}
		\begin{small}
			$$M_t\leq \frac{3^{m}\!+\!1}{4}+\sum_{j=t}^{m-1}2\cdot 3^j=\frac{3^{m+1}\!+\!2\!\cdot \!3^m\!-\!3^{t+1}\!-\!3^{t}\!+\!1}{4}\leq \frac{3^{m+1}\!+\!2\!\cdot\! 3^m\!-\!3^{2\ell+3}\!-\!3^{2\ell+2}\!+\!1}{4}<4n.$$
		\end{small}
		If $0<M_t<2n$, then $[a\cdot 3^{t}]_{2n}=M_t\geq \frac{3^m-3^{m-2\ell-2}-3^{m-2\ell-3}+5}{4}>\frac{3^{2\ell+1}-7}{4}$. If $2n<M_t<4n$, then
		\begin{small}
			\begin{align*}
				[a\cdot 3^{t}]_{2n}&=M_t-2n=\sum_{j=t}^{m-1}a_{j-t}3^{j}-\frac{3^{m+1}+3}{4}-\sum_{j=0}^{t-2\ell-1}a_{j+m-t}3^j\\
				&=\Big(\sum_{j=t}^{m-1}\big(a_{j-t}\!+\!(-1)^{j}\big)3^{j-t}\!-\!3^{m-t}\Big)3^t\!-\!\frac{3^t\!+\!3}{4}\!-\!\sum_{j=0}^{t-2\ell-1}\!a_{j+m-t}3^j\geq 3^t\!-\!\frac{3^t\!+\!3}{4}\!-\!\sum_{j=0}^{t-2\ell-1}\!2\cdot 3^j\\
				&=\frac{(3^{2\ell+1}-4)3^{t-2\ell}+1}{4}\geq \frac{(3^{2\ell+1}-4)\cdot 3^2+1}{4}=\frac{3^{2\ell+3}-35}{4}>\frac{3^{2\ell+1}-7}{4}.
			\end{align*}
		\end{small}
		
		{\bf Case 4.} $t=m-1$. Then $a\cdot 3^t\equiv \frac{3^m+1}{4}+a_0\cdot 3^{m-1}-\sum_{j=0}^{m-2\ell-2}a_{j+1}\cdot 3^j\triangleq M_t~({\rm mod}~2n)$.
		As $0<M_t\leq \frac{3^m+1}{4}+2\cdot 3^{m-1}=\frac{3^{m+1}+2\cdot 3^{m-1}+1}{4}<2n$, we have
		\begin{small}
			$$[a\cdot 3^t]_{2n}=M_t\geq \frac{3^m+1}{4}-\sum_{j=0}^{m-2\ell-2}2\cdot 3^j=\frac{3^m-3^{m-2\ell}-3^{m-2\ell-1}+5}{4}>\frac{3^{2\ell+1}-7}{4}.$$
		\end{small}
		
		{\bf Case 5.} $t=m$. In this case,
		\begin{small}
			$$[a\cdot 3^t]_{2n}=2n-a\geq 2n-\frac{3^m+3^{m-2\ell+1}}{4}+2=\frac{3^{m+1}-3^{m-2\ell+1}+12}{4}>\frac{3^{2\ell+1}-7}{4}.$$
		\end{small}
		
		{\bf Case 6.} $m+1\leq t\leq m+2\ell$. Let $t'=t-m$, then $1\leq t'\leq 2\ell$ and $[a\cdot 3^t]_{2n}=2n-[a\cdot 3^{t'}]_{2n}$. If $t$ is even, then $t'$ is odd and we see from Case 2 that
		\begin{small}
			$$[a\cdot 3^{t}]_{2n}\geq 2n-3^m+\frac{3^{2\ell+1}-3}{4}=\frac{3^{2\ell+1}+1}{4}.$$
		\end{small}
		If $t$ is odd, then $t'$ is even and from Case 2 we have
		\begin{small}
			$$[a\cdot 3^{t}]_{2n}\geq 2n-3^m+\frac{3^{2\ell+2}-1}{4}=\frac{3^{2\ell+2}+3}{4}>\frac{3^{2\ell+1}-7}{4}.$$
		\end{small}
		
		{\bf Case 7.} $m+2\ell+1\leq t\leq 2m-2$. Let $t'=t-m$, then $2\ell+1\leq t'\leq m-2$. Suppose $t$ is even, then $t'$ is odd and it follows from Case 3 that
		\begin{small}
			$$a\cdot 3^t\equiv 2n-a\cdot 3^{t'}\equiv \frac{3^m+1}{4}-\sum_{j=t'}^{m-1}a_{j-t'}3^{j}+\sum_{j=0}^{t'-2\ell-1}a_{j+m-t'}3^{j}\triangleq M_t~({\rm mod}~2n).$$
		\end{small}
		Note that
		\begin{small}
			$$M_t\geq \frac{3^m+1}{4}-\sum_{j=t'}^{m-1}2\cdot 3^{j}=-\frac{3^{m+1}-3^{t'+1}-3^{t}-1}{4}\geq -\frac{3^{m+1}-3^{2\ell+2}-3^{2\ell+1}-1}{4}>-2n,~{\rm and}$$
			$$M_t\leq \frac{3^m+1}{4}+\sum_{j=0}^{t'-2\ell-1}2\cdot 3^{j}=\frac{3^m+3^{m-2\ell-1}+3^{m-2\ell-2}-3}{4}<2n.$$
		\end{small}
		If $-2n<M_t<0$, then 
		\begin{small}
			$$[a\cdot 3^t]_{2n}=2n+M_t\geq 2n-\frac{3^{m+1}-3^{2\ell+2}-3^{2\ell+1}-1}{4}=\frac{3^{m}+3^{2\ell+2}+3^{2\ell+1}+5}{4}>\frac{3^{2\ell+1}-7}{4}.$$
		\end{small}
		If $0<M_t<2n$, then
		\begin{small}
			$$[a\cdot 3^t]_{2n}=M_t=\sum_{j=t'}^{m-1}\big((-1)^j-a_{j-t'}\big)3^j+\frac{3^{t'}+1}{4}+\sum_{j=0}^{t'-2\ell-1}a_{j+m-t'}3^{j}.$$
		\end{small}
		In this subcase, we claim that $\sum_{j=t'}^{m-1}\big((-1)^j-a_{j-t'}\big)3^j\geq 0$; otherwise, 
		\begin{small}
			$$M_t\leq -3^{t'}+\frac{3^{t'}+1}{4}+\sum_{j=0}^{t'-2\ell-1}2\cdot 3^{j}=-\frac{3^{t'+1}-3^{t'-2\ell+1}-3^{t'-2\ell}+3}{4}<0,$$
		\end{small}
		contradicting the fact that $0<M_t<2n$. It follows that $[b\cdot 3^t]_{2n}\geq \frac{3^{t'}+1}{4}>\frac{3^{2\ell+1}-7}{4}$.
		
		Suppose $t$ is odd, then $t'$ is even and we see from Case 3 that
		\begin{small}
			$$a\cdot 3^t\equiv 2n-a\cdot 3^{t'}\equiv \frac{3^{m+1}+3}{4}-\sum_{j=t'}^{m-1}a_{j-t'}3^{j}+\sum_{j=0}^{t'-2\ell-1}a_{j+m-t'}3^{j}\triangleq M_t~({\rm mod}~2n).$$
		\end{small}
		It is easy to check that
		\begin{small}
			$$M_t\geq \frac{3^{m+1}\!+\!3}{4}-\!\sum_{j=t'}^{m-1}\!2\cdot 3^{j}=-\frac{3^{m+1}\!-\!3^{t'+1}\!-\!3^{t'}\!-\!3}{4}=-\frac{3^{m+1}\!-\!3^{2\ell+3}\!-\!3^{2\ell+2}\!-\!3}{4}>-2n,~{\rm and}$$
		\end{small}
		\begin{small}
			$$M_t\leq \frac{3^{m+1}\!+\!3}{4}+\!\sum_{j=0}^{t'-2\ell-1}\!2\cdot 3^{j}=\frac{3^{m+1}\!+\!3^{t'-2\ell+1}\!+\!3^{t'-2\ell}\!-\!1}{4}\leq \frac{3^{m+1}\!+\!3^{m-2\ell-2}\!+\!3^{m-2\ell-3}\!-\!1}{4}<2n.$$
		\end{small}
		If $-2n<M_t<0$, then
		\begin{small}
			$$[a\cdot 3^t]_{2n}=2n+M_t\geq 2n-\frac{3^{m+1}-3^{2\ell+3}-3^{2\ell+2}-3}{4}=\frac{3^{m+1}+3^{2\ell+3}+3^{2\ell+2}+7}{4}>\frac{3^{2\ell+1}-7}{4}.$$
		\end{small}
		If $0<M_t<2n$, then
		\begin{small}
			$$[a\cdot 3^t]_{2n}=M_t=\Big(3^{m-t'}+\sum_{j=t'}^{m-1}\big((-1)^{j-1}-a_{j-t'}\big)3^{j-t'}\Big)3^{t'}+\frac{3^{t'}+3}{4}+\sum_{j=0}^{t'-2\ell-1}a_{j+m-t'}3^{j}.$$
		\end{small}
		In this subcase, we assert that $3^{m-t'}+\sum_{j=t'}^{m-1}\big((-1)^{j-1}-b_{j-t'}\big)3^{j-t'}\geq 0$; otherwise,
		\begin{small}
			$$M_t\leq -3^{t'}+\frac{3^{t'}+3}{4}+\sum_{j=0}^{t'-2\ell-1}2\cdot 3^{j}=-\frac{3^{t'+1}-3^{t'-2\ell+1}-3^{t'-2\ell}+1}{4}<0,$$
		\end{small}
		contradicting $0<M_t<2n$. Consequently, $[a\cdot 3^t]_{2n}\geq \frac{3^{t'}+3}{4}\geq \frac{3^{2\ell+2}+3}{4}>\frac{3^{2\ell+1}-7}{4}$.
		
		{\bf Case 8.} $t=2m-1$. It follows from Case 4 that
		\begin{small}
			\begin{align*}
				[a\cdot 3^t]_{2n}&=2n-[a\cdot 3^{m-1}]_{2n}=\frac{3^{m+1}+3}{4}-a_0\cdot 3^{m-1}+\sum_{j=0}^{m-2\ell-2}a_{j+1}3^j\\
				&\geq \frac{3^{m+1}+3}{4}-2\cdot 3^{m-1}=\frac{3^{m-1}+3}{4}>\frac{3^{2\ell+1}-7}{4}.
			\end{align*}
		\end{small}
		
		Summarizing all the discussions above, we get the desired result.
	\end{proof}
	
	Combining Lemmas \ref{l2.1}, \ref{l3.1} and \ref{l3.2}, we have the following lower bound on $d(\mathcal{C}_{(q,n,-1,\delta,1)}^{\perp})$.
	\begin{Theorem}\label{t3.1}
		Suppose $q=3$ and $m\geq 3$ is odd. For $2\leq \delta< \frac{\phi_1+3}{2}$, we have
		\begin{align*}
			d(\mathcal{C}_{(q,n,-1,\delta,1)}^{\perp})\geq \begin{cases}
				\frac{3^{m-1}+1}{2} & {\rm if}~\delta=2~{\rm or}~3;\\
				2\cdot 3^{m-2\ell-1}, & {\rm if}~\delta = \frac{3^{2\ell+1}+5}{8}~(1\leq \ell \leq \frac{m-3}{2});\\
				\frac{3^{m-2\ell}-1}{2}, & {\rm if}~\frac{3^{2\ell+1}+13}{8}\leq \delta \leq \frac{3^{2\ell+2}-1}{8}~(1\leq \ell \leq \frac{m-3}{2});\\
				3^{m-2\ell}, & {\rm if}~\frac{3^{2\ell}+7}{8}\leq \delta \leq \frac{3^{2\ell+1}-3}{8}~(2\leq \ell \leq \frac{m-3}{2});\\
				3, & {\rm if}~\frac{3^{m-1}+7}{8}\leq \delta \leq \frac{19\cdot 3^{m-3}+5}{8};\\
				2, & {\rm if}~\frac{19\cdot 3^{m-3}+13}{8}\leq \delta < \frac{\phi_1+3}{2}.
			\end{cases}
		\end{align*}
	\end{Theorem}
	\begin{proof}
		When $\delta=2~{\rm or}~3$, $T^{\perp}=(1+2\mathbb{Z}_{2n})\setminus C_1^{(q,2n)}$, where $C_1^{(q,2n)}=\{1,3,3^2,\cdots,3^{m-1},2n-3^{m-1},2n-3^{m-2},\cdots,2n-3,2n-1\}$. Then $\{3^{m-1}+2, 3^{m-1}+4,\cdots, 2n-3^{m-1}-2\}$ is a subset of $T^{\perp}$ containing $\frac{3^{m-1}-1}{2}$ elements. By Lemma \ref{l2.1},  $d(\mathcal{C}_{(q,n,-1,\delta,1)}^{\perp})\geq \frac{3^{m-1}+1}{2}$.
		
		When $4\leq \delta< \frac{\phi_2+3}{2}$, $\{I_1(\delta)+2,I_1(\delta)+4,\cdots,I_2(\delta)-2\}$ is a subset of $T^{\perp}$ containing $\frac{I_2(\delta)-I_1(\delta)}{2}-1$ elements. The desired result follows from Lemmas \ref{l2.1} \ref{l3.1} and \ref{l3.2}.
		
		When $\frac{\phi_2+3}{2}\leq \delta<\frac{\phi_1+3}{2}$, $\phi_2\leq 1+2(\delta-2)<\phi_1$. Then $T^{\perp}=C_{\phi_1}^{(q,2n)}$ and  $d(\mathcal{C}_{(q,n,-1,\delta,1)}^{\perp})\geq 2$.
	\end{proof}
	
	\begin{Example}{\rm 
			Take $q=3$ and $m=3$. By Theorem \ref{t3.1}, the lower bounds on the minimum distances of $\mathcal{C}_{(q,n,-1,2,1)}^{\perp}$ and $\mathcal{C}_{(q,n,-1,4,1)}^{\perp}$ are $5$ and $2$, respectively. By Magma, the true minimum distances of $\mathcal{C}_{(q,n,-1,2,1)}^{\perp}$ and $\mathcal{C}_{(q,n,-1,4,1)}^{\perp}$ are  $6$ and $2$, respectively.}
	\end{Example}
	
	We now present a sufficient and necessary condition for $\mathcal{C}_{(q,n,-1,\delta,1)}$ being a negacyclic dually-BCH code. The next lemma is important.
	\begin{Lemma}\label{l3.3}
		Suppose $q=3$ and $m>3$ is odd. For $4\leq \delta< \frac{\phi_2+3}{2}$, let $M(\delta)=T^{\perp}\setminus \big(C_{I_1(\delta)+2}^{(q,2n)}\cup C_{I_1(\delta)+4}^{(q,2n)}\cup \cdots\cup C_{I_2(\delta)-2}^{(q,2n)}\big)$. Then $\frac{3^{m-1}+1}{2}\in M(\delta)$ if $4\leq \delta \leq \frac{3^{m-1}+3}{4}$ and $\phi_2\in M(\delta)$ if $\frac{3^{m-1}+7}{4}\leq \delta< \frac{\phi_2+3}{2}$.
	\end{Lemma}
	
	\begin{proof}
		When $4\leq \delta \leq \frac{3^{m-1}+3}{4}$, $5\leq 1+2(\delta-2)\leq \frac{3^{m-1}-3}{2}$. By Lemmas \ref{l3.1} and \ref{l3.2}, ${\rm min}\{I_1(\delta)\mid 4\leq \delta \leq \frac{3^{m-1}+3}{4}\}=3^{m-2}+2\cdot 3^{m-3}$ and ${\rm max}\{I_2(\delta)\mid 4\leq \delta \leq \frac{3^{m-1}+3}{4}\}=3^{m-1}$. To prove $a:=\frac{3^{m-1}+1}{2}\in M(\delta)$, we need to prove that for every integer $t$ with $0\leq t\leq 2m-1$, $[a\cdot 3^t]_{2n}\in [\frac{3^{m-1}+1}{2}, 3^{m-2}+2\cdot 3^{m-3}]\cup [3^{m-1}, 2n-1]$. There are six cases.
		
		{\bf Case I.1.} $t=0$. $[a\cdot 3^t]_{2n}=a\in [\frac{3^{m-1}+1}{2}, 3^{m-2}+2\cdot 3^{m-3}]$.
		
		{\bf Case I.2.} $t=1$. $[a\cdot 3^t]_{2n}=a\cdot 3=\frac{3^m+3}{2}>3^{m-1}$.
		
		{\bf Case I.3.} $2\leq t\leq m-1$. $[a\cdot 3^t]_{2n}=\frac{3^m+2\cdot 3^{t-1}+1}{2}\geq \frac{3^{m}+7}{2}>3^{m-1}$.
		
		{\bf Case I.4.} $t=m$. $[a\cdot 3^t]_{2n}=2n-a=\frac{3^m+2\cdot 3^{m-1}+1}{2}>3^{m-1}$.
		
		{\bf Case I.5.} $t=m+1$. $[a\cdot 3^t]_{2n}=2n-a\cdot 3=\frac{3^m-1}{2}>3^{m-1}$.
		
		{\bf Case I.6.} $m+2\leq t\leq 2m-1$. In this case,
		\begin{small}
			$$[a\cdot 3^t]_{2n}=2n-[a\cdot 3^{t-m}]_{2n}=\frac{3^m-2\cdot 3^{t-m-1}-1}{2}\geq \frac{3^m-2\cdot 3^{m-2}-1}{2}>3^{m-1}.$$
		\end{small}
		
		When $\frac{3^{m-1}+7}{4}\leq \delta< \frac{\phi_2+3}{2}$, $\frac{3^{m-1}+1}{2}\leq 1+2(\delta-2)< \phi_2=\frac{7\cdot 3^{m-2}-1}{4}$. By Lemmas \ref{l3.1} and \ref{l3.2}, ${\rm min}\{I_1(\delta)\mid \frac{3^{m-1}+7}{4}\leq \delta< \frac{\phi_2+3}{2}\}=\frac{3^m-15}{4}$ and ${\rm max}\{I_2(\delta)\mid \frac{3^{m-1}+7}{4}\leq \delta< \frac{\phi_2+3}{2}\}=\frac{3^m+9}{4}$. Obviously, $\phi_2\in T^{\perp}$. It suffices to prove that either $[\phi_2 \cdot 3^t]_{2n}\leq \frac{3^m-15}{4}$ or $[\phi_2 \cdot 3^t]_{2n}\geq \frac{3^m+9}{4}$, where $0\leq t\leq 2m-1$. We have the following six cases.
		
		{\bf Case II.1.} $t=0$. $[\phi_2\cdot 3^t]_{2n}=\phi_2<\frac{3^m-15}{4}$.
		
		{\bf Case II.2.} $t=1$. $[\phi_2\cdot 3^t]_{2n}=\phi_2\cdot 3=\frac{2\cdot 3^{m}+3^{m-1}-3}{4}>\frac{3^m+9}{4}$.
		
		{\bf Case II.3.} $2\leq t\leq m-1$. Suppose $t$ is even, then
		\begin{small}
			$$\phi_2\cdot 3^t=\Big(\frac{2\!\cdot\! 3^{m-1}\!+\!3^{m-2}\!-\!3^{m-t+1}}{4}+\frac{3^{m-t+1}\!-\!1}{4}\Big)3^t\equiv \frac{3^{m+1}\!-\!3^t\!-\!2\!\cdot\! 3^{t-1}\!-\!3^{t-2}\!+\!3}{4}~({\rm mod}~2n).$$
		\end{small}
		One checks that $[\phi_2\cdot 3^t]_{2n}=\frac{3^{m+1}-3^t-2\cdot 3^{t-1}-3^{t-2}+3}{4}\geq \frac{3^{m+1}-3^{m-1}-2\cdot 3^{m-2}-3^{m-3}+3}{4}>\frac{3^m+9}{4}$. 
		
		Suppose $t$ is odd, then
		\begin{small}
			$$\phi_2\cdot 3^t=\Big(\frac{2\cdot 3^{m-1}+3^{m-2}-3^{m-t}}{4}+\frac{3^{m-t}-1}{4}\Big)3^t\equiv \frac{3^{m}-3^t-2\cdot 3^{t-1}-3^{t-2}+1}{4}~({\rm mod}~2n),$$
		\end{small}
		One can check that $[\phi_2\cdot 3^t]_{2n}=\frac{3^{m}-3^t-2\cdot 3^{t-1}-3^{t-2}+1}{4}\leq \frac{3^{m}-3^3-2\cdot 3^2-3+1}{4}=\frac{3^m-47}{4}<\frac{3^m-15}{4}$.
		
		{\bf Case II.4.} $t=m$. $[\phi_2\cdot 3^t]_{2n}=2n-\phi_2=\frac{3^{m+1}+2\cdot 3^{m-2}+5}{4}>\frac{3^m+9}{4}$.
		
		{\bf Case II.5.} $t=m+1$. $[\phi_2\cdot 3^t]_{2n}=2n-\phi_2\cdot 3=\frac{3^{m}+2\cdot 3^{m-1}+7}{4}>\frac{3^m+9}{4}$.
		
		{\bf Case II.6.} $m+2\leq t\leq 2m-1$. Let $t'=t-m$, then $2\leq t'\leq m-1$ and $[\phi_2\cdot 3^t]_{2n}=2n-[\phi_2\cdot 3^{t'}]_{2n}$. If $t$ is odd, then $t'$ is even and from Case 3 we have
		\begin{small}
			$$[\phi_2\cdot 3^t]_{2n}=\frac{3^m+3^{t'}+2\cdot 3^{t'-1}+3^{t'-2}+1}{4}\geq \frac{3^m+3^2+2\cdot 3+1+1}{4}=\frac{3^m+17}{4}>\frac{3^m+9}{4}.$$
		\end{small}
		If $t$ is even, then $t'$ is odd and from Case 3, $[\phi_2\cdot 3^t]_{2n}=\frac{3^{m+1}+3^{t'}+2\cdot 3^{t'-1}+3^{t'-2}-1}{4}>\frac{3^m+9}{4}$.
		
		We complete the proof.
	\end{proof}
	
	From Lemma \ref{l3.3}, we can obtain the following result.
	\begin{Theorem}\label{t3.2}
		Suppose $q=3$ and $m\geq 3$ is odd. For $2\leq \delta < \frac{\phi_1+3}{2}$, $\mathcal{C}_{(q,n,-1,\delta,1)}$ is a negacyclic dually-BCH code if and only if
		$$\delta=2,~3,~{\rm or}~\frac{\phi_2+3}{2}\leq \delta <\frac{\phi_1+3}{2}.$$
	\end{Theorem}
	
	\begin{proof}
		For $2\leq \delta < \frac{\phi_1+3}{2}$, $1\leq 1+2(\delta-2)\leq \phi_1$. It is clear that $\phi_1\in T^{\perp}$. Furthermore, we have $C_{\phi_1}^{(q,2n)}=\{\phi_1,2n-\phi_1\}$ and $C_j^{(q,2n)}=C_{2n-j}^{(q,2n)}$ for $j\in 1+2\mathbb{Z}_{2n}$, which means that $\mathcal{C}_{(q,n,-1,\delta,1)}$ is a negacyclic dually-BCH code if and only if there are two odd integers $J_1$ and $J_2$ with $J_1\leq \phi_1\leq J_2$ such that $T^{\perp}=C_{J_1}^{(q,2n)}\cup C_{J_1+2}^{(q,2n)}\cup\cdots \cup C_{\phi_1}^{(q,2n)}\cup C_{\phi_1+2}^{(q,2n)}\cup \cdots \cup C_{J_2}^{(q,2n)}$.

		When $\delta=2~{\rm or}~3$, $T^{\perp}=(1+2\mathbb{Z}_{2n})\setminus C_1^{(q,2n)}$. By Lemmas \ref{l3.1} and \ref{l3.2}, $I_1(\delta)=3^{m-2}$ and $I_2(\delta)=3^{m-1}$. We claim that $T^{\perp}=C_{3^{m-2}+2}^{(q,2n)}\cup C_{3^{m-2}+4}^{(q,2n)}\cup \cdots\cup C_{3^{m-1}-2}^{(q,2n)}$ and $\mathcal{C}_{(q,n,-1,\delta,1)}^{\perp}=\mathcal{C}_{(q,n,-1,3^{m-2},3^{m-2}+2)}$ is a negacyclic BCH code with designed distance $3^{m-2}$ with respect to $\beta$. To prove this, we need to prove that if $a\in T^{\perp}$, then $a\in C_{3^{m-2}+2}^{(q,2n)}\cup C_{3^{m-2}+4}^{(q,2n)}\cup \cdots\cup C_{3^{m-1}-2}^{(q,2n)}$, i.e., there is an integer $t$ with $0\leq t\leq 2m-1$ such that $3^{m-2}<[a\cdot 3^t]_{2n}<3^{m-1}$. Note that $C_{1}^{(q,2n)}=\{1,3,3^2,\cdots,3^{m-1},2n-3^{m-1},2n-3^{m-2},\cdots,2n-3,2n-1\}$ and that $C_j^{(q,2n)}=C_{2n-j}^{(q,2n)}$ for $j\in 1+2\mathbb{Z}_{2n}$; hence we only need consider odd integer $a$ satisfying $3^i<a<3^{i+1}$, where $0\leq i\leq m-2$, and $3^{m-1}<a\leq n-1$.
		
		If $3^i<a<3^{i+1}$ with $0\leq i\leq m-2$, then $3^{m-2}<a\cdot 3^{m-2-i}<3^{m-1}$.
		
		If $3^{m-1}<a\leq n-1=\sum_{j=0}^{m-1}3^j$, let $a=3^{m-1}+\sum_{j=0}^{m-2}a_j3^j$, where $a_j\in \{0,1,2\}$ for $0\leq j\leq m-2$. We have the following three cases.
		
		{\bf Case 1.} Suppose $(a_{\ell},a_{\ell-1})=(0,1)$ for some integer $\ell$ with $1\leq \ell\leq m-2$, i.e., $a=3^{m-1}+\sum_{j=\ell+1}^{m-2}a_j3^j+3^{\ell-1}+\sum_{j=0}^{\ell-2}a_j3^j$.
		
		If $\ell=1$ or $a_j=0$ for every integer $j$ with $0\leq j\leq \ell-2$, then $[a\cdot 3^{m-\ell}]_{2n}=3^{m-1}-3^{m-\ell-1}-\sum_{j=\ell+1}^{m-2}a_j3^{j-\ell}$. It is easy to see that
		\begin{small}
			$$3^{m-2}<3^{m-2}+3\leq  3^{m-1}-2\cdot 3^{m-\ell-1}+3\leq  [a\cdot 3^{m-\ell}]_{2n}\leq 3^{m-1}-3^{m-\ell-1}\leq 3^{m-1}-3<3^{m-1}.$$
		\end{small}
		
		If $a_j\neq 0$ for some integer $j$ with $0\leq j\leq \ell-2$, then $[a\cdot 3^{m-\ell-1}]_{2n}=3^{m-2}+\sum_{j=0}^{\ell-2}a_j3^{j+m-\ell-1}-3^{m-\ell-2}-\sum_{j=\ell+1}^{m-2}a_j3^{j-\ell-1}$. In addition,
		\begin{small}
			$$[a\cdot 3^{m-\ell-1}]_{2n}\geq 3^{m-2}\!+3^{m-\ell-1}\!-3^{m-\ell-2}\!-\!\sum_{j=\ell+1}^{m-2}\!2\cdot 3^{j-\ell-1}=3^{m-2}\!+3^{m-\ell-2}\!+\!1> 3^{m-2},~{\rm and}$$
			$$[a\cdot 3^{m-\ell-1}]_{2n}\leq 2\cdot 3^{m-2}-3^{m-\ell-1}-3^{m-\ell-2}\leq 2\cdot 3^{m-2}-4<3^{m-1}.$$
		\end{small}
		
		{\bf Case 2.} Suppose $(a_{\ell},a_{\ell-1})=(0,2)$ for some integer $\ell$ with $1\leq \ell\leq m-2$ , i.e., $a=3^{m-1}+\sum_{j=\ell+1}^{m-2}a_j3^j+2\cdot 3^{\ell-1}+\sum_{j=0}^{\ell-2}a_j3^j$. It is easy to check that $[a\cdot 3^{m-\ell-1}]_{2n}=2\cdot 3^{m-2}+\sum_{j=0}^{\ell-2}a_j3^{j+m-\ell-1}-3^{m-\ell-2}-\sum_{j=\ell+1}^{m-2}a_j3^{j-\ell-1}$. Moreover,
		\begin{small}
			$$[a\cdot 3^{m-\ell-1}]_{2n}\geq 2\cdot 3^{m-2}-2\cdot 3^{m-\ell-2}+1\geq 3^{m-2}+3^{m-3}+1>3^{m-2},~{\rm and}$$
			$$[a\cdot 3^{m-\ell-1}]_{2n}\leq 3^{m-1}-3^{m-\ell-1}-3^{m-\ell-2}\leq 3^{m-1}-4<3^{m-1}.$$
		\end{small}
		
		{\bf Case 3.} Suppose there is no integer $\ell$ with $1\leq \ell\leq m-2$ such that $(a_{\ell},a_{\ell-1})=(0,1)~{\rm or}~(0,2)$, then $a=\sum_{j=\ell}^{m-1}3^j$, where $0\leq \ell\leq m-3$ as $a$ is odd. It follows that $[a\cdot 3^{m-\ell-1}]_{2n}=3^{m-1}-\frac{3^{m-\ell-1}-1}{2}$. Furthermore, we have
		\begin{small}
			$$3^{m-2}<\frac{3^{m-1}+1}{2}=3^{m-1}-\frac{3^{m-1}-1}{2}\leq [a\cdot 3^{m-\ell-1}]_{2n}\leq 3^{m-1}-\frac{3^2-1}{2}=3^{m-1}-4<3^{m-1}.$$
		\end{small}
		
		We then conclude that $\mathcal{C}_{(q,n,-1,\delta,1)}$ is a negacyclic dually-BCH code when $\delta=2~{\rm or}~3$.
		
		When $4\leq \delta< \frac{\phi_2+3}{2}$, we know from Lemma \ref{l3.3} that $C_{I_1(\delta)+2}^{(q,2n)}\cup C_{I_1(\delta)+4}^{(q,2n)}\cup \cdots\cup C_{I_2(\delta)-2}^{(q,2n)}\subsetneq T^{\perp}$, which implies that there are no odd integers $J_1$ and $J_2$ with $J_1\leq \phi_1\leq J_2$ such that $T^{\perp}=C_{J_1}^{(q,2n)}\cup C_{J_1+2}^{(q,2n)}\cup\cdots \cup C_{J_2}^{(q,2n)}$, i.e., $\mathcal{C}_{(q,n,-1,\delta,1)}$ is not a negacyclic dually-BCH code.
		
		When $\frac{\phi_2+3}{2}\leq \delta <\frac{\phi_1+3}{2}$, $T^{\perp}=C_{\phi_1}^{(q,2n)}$ and $\mathcal{C}_{(q,n,-1,\delta,1)}^{\perp}=\mathcal{C}_{(q,n,-1,2,\phi_1)}$ with respect to $\beta$.
	\end{proof}
	
	\begin{Example}{\rm 
			Take $q=3$ and $m=3$. By Magma, for every integer $\delta$ with $2\leq \delta<\frac{\phi_1+3}{2}=5$, $\mathcal{C}_{(q,n,-1,\delta,1)}$ is a negacyclic dually-BCH code, which coincides with Theorem \ref{t3.2}.}
	\end{Example}
	
	\begin{Example}{\rm 
			Take $q=3$ and $m=5$. By Magma, for $2\leq \delta<\frac{\phi_1+3}{2}=32$, $\mathcal{C}_{(q,n,-1,\delta,1)}$ is a negacyclic dually-BCH code if and only if $\delta=2,~3$, or $25\leq \delta <32$, which coincides with Theorem \ref{t3.2}.}
	\end{Example}

	\noindent {\bf A.2. The subcase $m$ is even}
	
	Assume that $q=3$ and $m$ is even in this subsection. We first give a lower bound on the minimum distance of $\mathcal{C}_{(q,n,-1,\delta,1)}^{\perp}$. By the same way as Lemma \ref{l3.3}, we have the following result.
	
	\begin{Lemma}\label{l3.4}
		Suppose $q=3$ and $m>2$ is even. For $2\leq \delta< \frac{\phi_1+3}{2}$, let $I(\delta)$ be the odd integer such that $I(\delta)\not\in T^{\perp}$ and $\{I(\delta)+2,I(\delta)+4,\cdots,\phi_1\}\in T^{\perp}$. Then
		\begin{align*}
			I(\delta)=\begin{cases}
				\frac{3^m-3^{m-2\ell+1}}{2}, & {\rm if}~\frac{3^{2\ell-1}+5}{4}\leq \delta \leq \frac{3^{2\ell}+3}{4};\\
				\frac{3^m-3^{m-2\ell}+2}{2}, & {\rm if}~\frac{3^{2\ell}+7}{4}\leq \delta \leq \frac{3^{2\ell+1}+1}{4},
			\end{cases}
			~~{\rm where}~1\leq \ell \leq \frac{m-2}{2}.
		\end{align*}
	\end{Lemma}
	
	\begin{Theorem}\label{t3.3}
		Suppose $q=3$ and $m>2$ is even. For $2\leq \delta< \frac{\phi_1+3}{2}$, we have
		\begin{align*}
			d(\mathcal{C}_{(q,n,-1,\delta,1)}^{\perp})\geq\begin{cases}
				\frac{3^{m-2\ell+1}+1}{2}, & {\rm if}~\frac{3^{2\ell-1}+5}{4}\leq \delta \leq \frac{3^{2\ell}+3}{4};\\
				\frac{3^{m-2\ell}-1}{2}, & {\rm if}~\frac{3^{2\ell}+7}{4}\leq \delta \leq \frac{3^{2\ell+1}+1}{4},
			\end{cases}
			~~{\rm where}~1\leq \ell \leq \frac{m-2}{2}.
		\end{align*}
	\end{Theorem}
	\begin{proof}
		For $2\leq \delta< \frac{\phi_1+3}{2}$, $\{I(\delta)+2, I(\delta)+4,\cdots,\phi_1,\phi_1+2,\cdots,2n-I(\delta)-2\}$ is a subset of $T^{\perp}$ containing $n-I(\delta)-1$ elements. The result follows from Lemmas \ref{l2.1} and \ref{l3.4}.
	\end{proof}
	
	\begin{Example}{\rm
			Take $q=3$ and $m=4$. By Theorem \ref{t3.3}, the lower bounds on the minimum distances of $\mathcal{C}_{(q,n,-1,2,1)}^{\perp}$ and $\mathcal{C}_{(q,n,-1,7,1)}^{\perp}$ are $22$ and $4$, respectively. By Magma, the true minimum distances of $\mathcal{C}_{(q,n,-1,2,1)}^{\perp}$ and $\mathcal{C}_{(q,n,-1,7,1)}^{\perp}$ are $23$ and $5$, respectively.}
	\end{Example}
	
	To investigate the condition for $\mathcal{C}_{(q,n,-1,\delta,1)}$ being a negacyclic dually-BCH code, we need the following lemma, whose proof we omit as it is similar to that of Lemma \ref{l3.3}.
	\begin{Lemma}\label{l3.5}
		Suppose $q=3$ and $m>2$ is even. For $2\leq \delta< \frac{\phi_3+3}{2}$, let $M(\delta)=T^{\perp}\setminus \big(C_{I(\delta)+2}^{(q,2n)}\cup C_{I(\delta)+4}^{(q,2n)}\cup \cdots\cup C_{\phi_1}^{(q,2n)}\big)$. Then $\frac{3^{m-1}+1}{4}\in M(\delta)$ if $2\leq \delta<\frac{3^{m-1}+13}{8}$ and $\phi_3\in M(\delta)$ if $\frac{3^{m-1}+13}{8}\leq \delta < \frac{\phi_3+3}{2}$.
	\end{Lemma}

	\begin{Theorem}\label{t3.4}
		Suppose $q=3$ and $m>2$ is even. For $2\leq \delta < \frac{\phi_1+3}{2}$, $\mathcal{C}_{(q,n,-1,\delta,1)}$ is a negacyclic dually-BCH code if and only if
		$$\frac{\phi_3+3}{2}\leq \delta <\frac{\phi_1+3}{2}.$$
	\end{Theorem}
	\begin{proof}
		It is obvious that $\phi_1\in T^{\perp}$. As $C_{\phi_1}^{(q,2n)}=\{\phi_1\}$ and $C_j^{(q,2n)}=C_{2n-j}^{(q,2n)}$ for $j\in 1+2\mathbb{Z}_{2n}$, $\mathcal{C}_{(q,n,-1,\delta,1)}$ is a negacyclic dually-BCH code if and only if there is an odd integer $J$ with $J\leq \phi_1$ such that $T^{\perp}=C_{J}^{(q,2n)}\cup C_{J+2}^{(q,2n)}\cup\cdots \cup C_{\phi_1}^{(q,2n)}$.
		
		When $2\leq \delta < \frac{\phi_1+3}{2}$, we see from Lemma \ref{l3.5} that $C_{I(\delta)+2}^{(q,2n)}\cup C_{I(\delta)+4}^{(q,2n)}\cup \cdots\cup C_{\phi_1}^{(q,2n)}\subsetneq T^{\perp}$, which implies that there is no odd integer $J$ with $J\leq \phi_1$ such that $T^{\perp}=C_{J}^{(q,2n)}\cup C_{J+2}^{(q,2n)}\cup\cdots \cup C_{\phi_1}^{(q,2n)}$, i.e., $\mathcal{C}_{(q,n,-1,\delta,1)}$ is not a negacyclic dually-BCH code.
		
		When $\delta=\frac{\phi_3+3}{2}$, $1+2(\delta-2)=\phi_3=\phi_2-2$; in this case, $T^{\perp}=C_{\phi_2}^{(q,2n)}\cup C_{\phi_1}^{(q,2n)}=C_{\frac{3^{m-1}-1}{2}}^{(q,2n)}\cup C_{\frac{3^m+1}{2}}^{(q,2n)}=C_{\frac{3^{m}-3}{2}}^{(q,2n)}\cup C_{\frac{3^m+1}{2}}^{(q,2n)}$ and $\mathcal{C}_{(q,n,-1,\delta,1)}^{\perp}=\mathcal{C}_{(q,n,-1,3,\frac{3^{m}-3}{2})}$ with respect to $\beta$.
		
		When $\frac{\phi_2+3}{2}\leq \delta <\frac{\phi_1+3}{2}$, $T^{\perp}=C_{\phi_1}$ and $\mathcal{C}_{(q,n,-1,\delta,1)}^{\perp}=\mathcal{C}_{(q,n,-1,2,\phi_1)}$ with respect to $\beta$.
	\end{proof}
	
	\begin{Example}{\rm 
			Take $q=3$ and $m=4$. By Magma, for $2\leq \delta<\frac{\phi_1+3}{2}=22$,  $\mathcal{C}_{(q,n,-1,\delta,1)}$ is a negacyclic dually-BCH code if and only if $7\leq \delta <22$, which coincides with Theorem \ref{t3.4}.}
	\end{Example}
	
	\subsection*{B. The case $q>3$}	
	Our task in this subsection is to give a lower bound on the minimum distance of the dual code $\mathcal{C}_{(q,n,-1,\delta,1)}^{\perp}$ and present a characterization of $\mathcal{C}_{(q,n,-1,\delta,1)}$ being a dually-BCH code. The cases $m$ is odd and $m$ is even should be treated separately.
	
	\subsection*{B.1. The subcase $m$ is odd}

	The following two lemmas will be employed later. 
	\begin{Lemma}\label{L3.1}
		Suppose $q>3$ and $m\geq 3$ is odd. For $2\leq \delta<\frac{\phi_2+3}{2}$, let $I_1(\delta)$ be the odd integer such that $I_1(\delta)\not\in T^{\perp}$ and $\{I_1(\delta)+2,I_1(\delta)+4,\cdots,\phi_1\}\subseteq T^{\perp}$. Then
		\begin{small}
			\begin{align*}
				I_1(\delta)\!=\!\begin{cases}
					(\Delta_{\ell_0,q}-\frac{q-3}{2}+2\ell_1)q^{m-2\ell_0+1}, & {\rm if}~\delta=\frac{2\Delta_{\ell_0,q}-q+9}{4}+\ell_1~(0\leq \ell_1<\frac{q-7}{4});\\
					(\Delta_{\ell_0,q}-2)q^{m-2\ell_0+1}, & {\rm if}~\frac{\Delta_{\ell_0,q}+1}{2}\leq \delta\leq \frac{\Delta_{\ell_0,q}'+1}{2};\\
					(q^{2\ell_0-1}-\Delta_{\ell_0,q}')q^{m-2\ell_0+1}\!+\!2\ell_1\!+\!1, & {\rm if}~\ell_1q^{2\ell_0-1}+\frac{\Delta_{\ell_0,q}'+3}{2}\leq \delta \leq (\ell_1+1)q^{2\ell_0-1}\\
					&+\frac{\Delta_{\ell_0,q}'+1}{2}~(0\leq \ell_1<\frac{q-7}{4});\\
					(q^{2\ell_0-1}-\Delta_{\ell_0,q}')q^{m-2\ell_0+1}+\frac{q-5}{2}, & {\rm if}~\frac{(q-7)q^{2\ell_0-1}+2\Delta_{\ell_0,q}'+6}{4}\leq \delta \leq \frac{q\Delta_{\ell_0,q}-q+2}{2};\\
					(q\Delta_{\ell_0,q}-q+2\ell_1+1)q^{m-2\ell_0}, & {\rm if}~\delta=\frac{q\Delta_{\ell_0,q}-q+4}{2}+\ell_1~(0\leq \ell_1<\frac{q-3}{4});\\
					(q\Delta_{\ell_0,q}-\frac{q+1}{2})q^{m-2\ell_0}, & {\rm if}~\frac{2q\Delta_{\ell_0,q}-q+5}{4}\leq \delta \leq \frac{2q\Delta_{\ell_0,q}'-q+1}{4};\\
					(q^{2\ell_0}\!-\!q\Delta_{\ell_0,q}'\!+\!\frac{q+1}{2})q^{m-2\ell_0}\!+\!2\ell_1\!+\!1, & {\rm if}~\ell_1q^{2\ell_0}+\frac{2q\Delta_{\ell_0,q}'-q+5}{4}\leq \delta \leq (\ell_1+1)q^{2\ell_0}\\
					&+\frac{2q\Delta_{\ell_0,q}'-q+1}{4}~(0\leq \ell_1<\frac{q-7}{4});\\
					(q^{2\ell_0}\!-\!q\Delta_{\ell_0,q}'\!+\!\frac{q+1}{2})q^{m-2\ell_0}\!+\!\frac{q-7}{2}, & {\rm if}~\frac{(q-7)q^{2\ell_0}+2q\Delta_{\ell_0,q}'-q+5}{4}\leq \delta \leq \frac{2\Delta_{\ell_0+1,q}-q+5}{4},
				\end{cases}
			\end{align*}
		\end{small}
		
		\noindent where $1\leq \ell_0\leq \frac{m-1}{2}$, $\Delta_{\ell_0,q}=\frac{(q-1)(q^{2\ell_0-1}+1)}{2(q+1)}$, $\Delta_{\ell_0,q}'=\frac{(q+3)(q^{2\ell_0-1}+1)}{2(q+1)}$, and $\Delta_{\ell_0+1,q}=\frac{(q-1)(q^{2\ell_0+1}+1)}{2(q+1)}$.
	\end{Lemma}
	
	\begin{proof}
		We only prove the first case, and the proof of other cases is similar. When $\delta=\frac{2\Delta_{\ell_0,q}-q+9}{4}+\ell_1$, $1+2(\delta-2)=\Delta_{\ell_0,q}-\frac{q-3}{2}+2\ell_1$. It is easy to see that $(\Delta_{\ell_0,q}-\frac{q-3}{2}+2\ell_1)q^{m-2\ell_0+1}\in C_{\Delta_{\ell_0,q}-\frac{q-3}{2}+2\ell_1}^{(q,2n)}\nsubseteq T^{\perp}=(1+2\mathbb{Z}_{2n})\setminus \big(C_1^{(q,2n)}\cup C_3^{(q,2n)}\cup \cdots\cup C_{1+2(\delta-2)}^{(q,2n)}\big)$. Now we are going to show that
		\begin{small}
			$$\Gamma\!:=\!\Big\{\!\big(\Delta_{\ell_0,q}-\frac{q\!-\!3}{2}+2\ell_1\big)q^{m-2\ell_0+1}+2,\big(\Delta_{\ell_0,q}-\frac{q\!-\!3}{2}+2\ell_1\big)q^{m-2\ell_0+1}+4,\cdots\!,\phi_1\Big\}\!\subseteq T^{\perp}.$$
		\end{small}
		For $a\in \Gamma$, $a=(\Delta_{\ell_0,q}-\frac{q-3}{2}+2\ell_1)q^{m-2\ell_0+1}+u$, where $u$ is an even integer satisfying $2\leq u\leq \phi_1-(\Delta_{\ell_0,q}-\frac{q-3}{2}+2\ell_1)q^{m-2\ell_0+1}=(\frac{q-5}{2}-2\ell_1)q^{m-2\ell_0+1}+\frac{q+1}{2}q^{m-2\ell_0}+\frac{(q-1)(q^{m-2\ell_0}+1)}{2(q+1)}$. Suppose $u=\sum_{j=0}^{m-2\ell_0+1}a_jq^j$, where $0\leq a_{m-2\ell_0+1}\leq \frac{q-5}{2}-2\ell_1$ and $0\leq a_j\leq q-1$ for $0\leq j\leq m-2\ell_0$. Then
		\begin{small}
			$$a=\frac{(q-1)(q^m-q^{m-2\ell_0+2})}{2(q+1)}+(2\ell_1+1)q^{m-2\ell_0+1}+\sum_{j=0}^{m-2\ell_0+1}a_jq^j.$$
		\end{small}
		To prove $a\in T^{\perp}$, we need to prove that $[aq^t]_{2n}>\Delta_{\ell_0,q}-\frac{q-3}{2}+2\ell_1$ for every integer $t$ with $0\leq t\leq 2m-1$. We organize the proof into the following four cases.

		{\bf Case 1.} $t=0$. It is obvious that $[aq^t]_{2n}=a>\Delta_{\ell_0,q}-\frac{q-3}{2}+2\ell_1$.
		
		{\bf Case 2.} $1\leq t\leq 2\ell_0-2$. Suppose $t$ is odd, then
		\begin{small}
			\begin{align*}
				aq^t&=\frac{q-1}{2}\Big(\frac{(q^t+1)q^{m-t}}{q+1}-\frac{q^{m-t}+q^{m-2\ell_0+2}}{q+1}\Big)q^t+(2\ell_1+1)q^{m-2\ell_0+1+t}+uq^t\\
				&\equiv \frac{q-1}{2}\Big(-\frac{q^t+1}{q+1}-\frac{q^m+q^{m-2\ell_0+2+t}}{q+1}\Big)+(2\ell_1+1)q^{m-2\ell_0+1+t}+uq^t\\
				&\equiv q^m+1-\frac{q-1}{2}\Big(\frac{q^m+1}{q+1}+\frac{q^{m-2\ell_0+2+t}+q^t}{q+1}\Big)+(2\ell_1+1)q^{m-2\ell_0+1+t}+uq^t\\
				&=\frac{(q+3)(q^m+1)}{2(q+1)}-\Big(\frac{(q-1)(q^{m-2\ell_0+2}+1)}{2(q+1)}-(2\ell_1+1)q^{m-2\ell_0+1}-u\Big)q^t\triangleq M_t~({\rm mod}~2n).
			\end{align*}
		\end{small}
		
		\noindent	It is easy to check that $0<M_t<\frac{(q+3)(q^m+1)}{2(q+1)}<2n$. Therefore,
		\begin{small}
			\begin{align*}
				[aq^t]_{2n}&=M_t\geq \frac{(q+3)(q^m+1)}{2(q+1)}-\Big(\frac{(q-1)(q^{m-2\ell_0+2}+1)}{2(q+1)}-q^{m-2\ell_0+1}-2\Big)q^{2\ell_0-3}\\
				&=\frac{(q+3)(q^m+1)}{2(q+1)}-\frac{(q-1)(q^{m-1}+q^{2\ell_0-3})}{2(q+1)}+q^{m-2}+2q^{2\ell_0-3}\\
				&>\frac{(q+3)(q^m+1)}{2(q+1)}-\frac{(q-1)(q^m+1)}{2(q+1)}=\frac{2(q^m+1)}{q+1}>\Delta_{\ell_0,q}-\frac{q-3}{2}+2\ell_1.
			\end{align*}
		\end{small}
		
		Suppose $t$ is even, then
		\begin{small}
			\begin{align*}
				aq^t&=\frac{q-1}{2}\Big(\frac{(q^t-1)q^{m-t}}{q+1}+\frac{q^{m-t}-q^{m-2\ell_0+2}}{q+1}\Big)q^t+(2\ell_1+1)q^{m-2\ell_0+1+t}+uq^t\\
				&\equiv \frac{q-1}{2}\Big(-\frac{q^t-1}{q+1}+\frac{q^m-q^{m-2\ell_0+2+t}}{q+1}\Big)+(2\ell_1+1)q^{m-2\ell_0+1+t}+uq^t\\
				&= \frac{(q-1)(q^m+1)}{2(q+1)}-\Big(\frac{(q-1)(q^{m-2\ell_0+2}+1)}{2(q+1)}-(2\ell_1+1)q^{m-2\ell_0+1}-u\Big)q^t\triangleq M_t~({\rm mod}~2n).
			\end{align*}
		\end{small}
		
		\noindent One can check that $0<M_t<\frac{(q-1)(q^m+1)}{2(q+1)}<2n$, thus
		\begin{small}
			\begin{align*}
				[aq^t]_{2n}&=M_t\geq \frac{(q-1)(q^m+1)}{2(q+1)}-\Big(\frac{(q-1)(q^{m-2\ell_0+2}+1)}{2(q+1)}-q^{m-2\ell_0+1}-2\Big)q^{2\ell_0-2}\\
				&=q^{m-1}+2q^{2\ell_0-2}-\frac{(q-1)(q^{2\ell_0-2}-1)}{2(q+1)}>q^{m-1}>\Delta_{\ell_0,q}-\frac{q-3}{2}+2\ell_1.
			\end{align*}
		\end{small}
		
		{\bf Case 3.} $2\ell_0-1\leq t\leq m-1$. In this case,
		\begin{small}
			\begin{align*}
				aq^t&=\Big(\frac{(q-1)(q^t-q^{t-2\ell_0+2})q^{m-t}}{2(q+1)}+(2\ell_1+1)q^{m-2\ell_0+1}+\sum_{j=0}^{m-t-1}a_jq^j+\sum_{j=m-t}^{m-2\ell_0+1}a_jq^j\Big)q^t\\
				&\equiv\! \sum_{j=0}^{m-t-1}\!a_jq^{j+t}-\frac{(q\!-\!1)(q^t\!-\!q^{t-2\ell_0+2})}{2(q+1)}-(2\ell_1\!+\!1)q^{t-2\ell_0+1}-\!\sum_{j=m-t}^{m-2\ell_0+1}\!a_jq^{j+t-m}\triangleq M_t~({\rm mod}~2n).
			\end{align*}
		\end{small}
		
		\noindent  If there is an integer $j$ with $0\leq j\leq m-t-1$ such that $a_j\neq 0$, then
		\begin{small}
			\begin{align*}
				M_t&\geq q^t-\frac{(q\!-\!1)(q^t\!-\!q^{t-2\ell_0+2})}{2(q+1)}-(2\ell_1\!+\!1)q^{t-2\ell_0+1}-(\frac{q\!-\!5}{2}-2\ell_1)q^{t-2\ell_0+1}-\sum_{j=m-t}^{m-2\ell_0}(q\!-\!1)q^{j+t-m}\\
				&=\big(q^{2\ell_0-1}-\frac{(q-1)(q^{2\ell_0-1}+1)}{2(q+1)}\big)q^{t-2\ell_0+1}+1>0,
			\end{align*}
		\end{small}
		
		\noindent 	and $M_t<\sum_{j=0}^{m-t-1}(q-1)q^{j+t}=q^m-q^t<2n$; consequently,
		\begin{small}
			$$	[aq^t]_{2n}=M_t\geq q^{2\ell_0-1}-\frac{(q-1)(q^{2\ell_0-1}+1)}{2(q+1)}+1>\frac{(q-1)(q^{2\ell_0-1}+1)}{2(q+1)}>\Delta_{\ell_0,q}-\frac{q-3}{2}+2\ell_1.$$
		\end{small}
		If $a_j=0$ for every integer $j$ with $0\leq j\leq m-t-1$, then 
		\begin{small}
			\begin{align*}
				[aq^t]_{2n}&=q^m+1-\frac{(q-1)(q^t-q^{t-2\ell_0+2})}{2(q+1)}-(2\ell_1+1)q^{t-2\ell_0+1}-\sum_{j=m-t}^{m-2\ell_0+1}a_jq^{j+t-m}\\
				&\geq q^m+1-\frac{(q-1)(q^t+q^{t-2\ell_0+1})}{2(q+1)}+1>q^m+1-\frac{(q-1)(q^m+1)}{2(q+1)}\\
				&=\frac{(q+3)(q^m+1)}{2(q+1)}>\Delta_{\ell_0,q}-\frac{q-3}{2}+2\ell_1.
			\end{align*}
		\end{small}
		
		{\bf Case 4.} $m\leq t\leq 2m-1$. Let $t'=t-m$, then $0\leq t'\leq m-1$ and $[aq^t]_{2n}=2n-[aq^{t'}]_{2n}$. It can be seen from Cases 1-3 that $[aq^t]_{2n}>\Delta_{\ell_0,q}-\frac{q-3}{2}+2\ell_1$. This completes the proof.
	\end{proof}
	
	\begin{Lemma}\label{L3.2}
		Suppose $q>3$ and $m\geq 3$ is odd. For $2\leq \delta< \frac{\phi_2+3}{2}$, let $I_2(\delta)$ be the odd integer such that $\{\phi_1,\phi_1+2,\cdots,I_2(\delta)-2\}\subseteq T^{\perp}$ and $I_2(\delta)\not\in T^{\perp}$. Then
		\begin{small}
			\begin{align*}
				I_2(\delta)\!=\!\begin{cases}
					(q-2\ell_1-1)q^{m-1}\!+\!1, & {\rm if}~\delta=\ell_1+2~(0\leq \ell_1\leq \frac{q-7}{4});\\
					\frac{q^m-q^{m-1}}{2}\!-\!2\ell_1, & {\rm if}~\ell_1q+\frac{q+5}{4}\leq \delta\leq (\ell_1+1)q+\frac{q+1}{4}~(0\leq \ell_1\leq \frac{q-7}{4});\\
					\frac{q^m-q^{m-1}-q+3}{2}, & {\rm if}~\delta=\frac{q^2-2q+5}{4};\\
					(q\Delta_{\ell_0,q}-\frac{q-3}{2})q^{m-2\ell_0}\!-\!2\ell_1, & {\rm if}~\ell_1q^{2\ell_0}+\frac{2q\Delta_{\ell_0,q}-q+9}{4}\leq \delta\leq (\ell_1+1)q^{2\ell_0}+\frac{2q\Delta_{\ell_0,q}-q+5}{4}\\
					&(0\leq \ell_1\leq \frac{q-7}{4});\\
					(q\Delta_{\ell_0,q}\!-\!\frac{q-3}{2})q^{m-2\ell_0}\!-\!\frac{q-3}{2}, & {\rm if}~\frac{(q-3)q^{2\ell_0}+2q\Delta_{\ell_0,q}-q+9}{4}\leq \delta\leq \frac{\Delta_{\ell_0+1,q}+1}{2};\\
					q^{m-2\ell_0-1}\Delta_{\ell_0+1,q}-2\ell_1, & {\rm if}~\ell_1q^{2\ell_0+1}+\frac{\Delta_{\ell_0+1,q}+3}{2}\leq \delta\leq (\ell_1+1)q^{2\ell_0+1}+\frac{\Delta_{\ell_0+1,q}+1}{2}\\
					&(0\leq \ell_1\leq \frac{q-7}{4});\\
					q^{m-2\ell_0-1}\Delta_{\ell_0+1,q}-\frac{q-3}{2}, & {\rm if}~\frac{(q-3)q^{2\ell_0+1}+2\Delta_{\ell_0+1,q}+6}{4}\leq \delta\leq \frac{2q\Delta_{\ell_0+1,q}-q+5}{4};\\
					\frac{q^{m+1}-q^m+q^2+3q}{2(q+1)}-2\ell_1, & {\rm if}~\ell_1q^{m-1}+\frac{q^m-q^{m-1}+7q+9}{4(q+1)}\leq \delta\leq (\ell_1+1)q^{m-1}\\
					&+\frac{q^m-q^{m-1}+3q+5}{4(q+1)}~(0\leq \ell_1<\frac{q-7}{4});\\
					\frac{q^{m+1}-q^m+9q+7}{2(q+1)},& {\rm if}~\frac{q^{m+1}-5q^m-8q^{m-1}+7q+9}{4(q+1)}\!\leq\! \delta\!\leq\! \frac{q^{m+1}-q^m-4q^{m-1}-4q^{m-2}+3q+1}{4(q+1)};\\
					\frac{q^{m+1}-q^m+5q+3}{2(q+1)}, & {\rm if}~\frac{q^{m+1}-q^m-4q^{m-1}-4q^{m-2}+7q+5}{4(q+1)}\leq \delta< \frac{\phi_2+3}{2},
				\end{cases}
			\end{align*}
		\end{small}
		
		\noindent where $1\leq \ell_0\leq \frac{m-3}{2}$, $\Delta_{\ell_0,q}=\frac{(q-1)(q^{2\ell_0-1}+1)}{2(q+1)}$ and $\Delta_{\ell_0+1,q}=\frac{(q-1)(q^{2\ell_0+1}+1)}{2(q+1)}$.
	\end{Lemma}
	\begin{proof}
		We are just going to prove the fourth case, as the others can be proved similarly.
		When $\delta$ is in the fourth case, we have $2\ell_1q^{2\ell_0}+q\Delta_{\ell_0,q}-\frac{q-3}{2}\leq 1+2(\delta-2)\leq f:=(2\ell_1+2)q^{2\ell_0}+q\Delta_{\ell_0,q}-\frac{q+1}{2}$. Then $(q\Delta_{\ell_0,q}-\frac{q-3}{2})q^{m-2\ell_0}-2\ell_1\notin T^{\perp}=(1+2\mathbb{Z}_{2n})\setminus \big(C_1^{(q,2n)}\cup C_3^{(q,2n)}\cup \cdots\cup C_{1+2(\delta-2)}^{(q,2n)}\big)$ as $\big((q\Delta_{\ell_0,q}-\frac{q-3}{2})q^{m-2\ell_0}-2\ell_1\big)q^{m+2\ell_0}=2\ell_1q^{2\ell_0}+q\Delta_{\ell_0,q}-\frac{q-3}{2}\notin T^{\perp}$. We next prove that
		\begin{small}
			$$\Gamma:=\Big\{\phi_1, \phi_1+2,\cdots, (q\Delta_{\ell_0,q}-\frac{q-3}{2})q^{m-2\ell_0}-2\ell_1-2\Big\}\subseteq T^{\perp}.$$
		\end{small}
		For $a\in \Gamma$, $a=\frac{(q-1)(q^m+1)}{2(q+1)}+u$, where $u$ is an even integer satisfying $0\leq u\leq (q\Delta_{\ell_0,q}-\frac{q-3}{2})q^{m-2\ell_0}-2\ell_1-2-\phi_1=\frac{q+1}{2}q^{m-2\ell_0-1}+\frac{q-3}{2}q^{m-2\ell_0-2}+\frac{(q+3)(q^{m-2\ell_0-2}+1)}{2(q+1)}-2\ell_1-3$. Suppose $u=\sum_{j=0}^{m-2\ell_0-1}a_jq^j$, where $0\leq a_{m-2\ell_0-1}\leq \frac{q+1}{2}$ and $0\leq a_j\leq q-1$ for $0\leq j\leq m-2\ell_0-2$. In order to prove $a\in T^{\perp}$, it suffices to prove $[aq^t]_{2n}>f$ for every integer $t$ with $0\leq t\leq 2m-1$. We divide the proof into four cases.
		
		{\bf Case 1.} $t=0$. Then $[aq^t]_{2n}=a\geq \frac{(q-1)(q^m+1)}{2(q+1)}>f$.
		
		{\bf Case 2.} $1\leq t\leq 2\ell_0$. Suppose $t$ is odd, then
		\begin{small}
			\begin{align*}
				aq^t=\Big(\frac{(q\!-\!1)(q^m\!+\!1)}{2(q+1)}+u\Big)q^t\equiv -\frac{(q\!-\!1)(q^m\!+\!1)}{2(q+1)}+uq^t\equiv \frac{(q\!+\!3)(q^m\!+\!1)}{2(q+1)}+uq^t\triangleq M_t~({\rm mod}~2n).
			\end{align*}
		\end{small}
		
		\noindent It is obvious that $M_t>0$, and
		\begin{small}
			\begin{align*}
				M_t&\leq \frac{(q+3)(q^m+1)}{2(q+1)}+\Big(\frac{(q+3)(q^{m-2\ell_0}+1)}{2(q+1)}-2\ell_1-3\Big)q^{2\ell_0-1}\\
				&=\frac{(q+3)q^{m-1}}{2}-(2\ell_1+3)q^{2\ell_0-1}+\frac{(q+3)(q^{2\ell_0-1}+1)}{2(q+1)}<\frac{(q+3)q^{m-1}}{2}<2n;
			\end{align*}
		\end{small}
		
		\noindent therefore, $[aq^t]_{2n}=M_t\geq \frac{(q+3)(q^m+1)}{2(q+1)}>f$. 
		
		Suppose $t$ is even, then
		\begin{small}
			$$aq^t=\Big(\frac{(q-1)(q^m+1)}{2(q+1)}+u\Big)q^t\equiv \frac{(q-1)(q^m+1)}{2(q+1)}+uq^t \triangleq M_t~({\rm mod}~2n).$$
		\end{small}
		Note that $M_t>0$, and
		\begin{small}
			\begin{align*}
				M_t&\leq \frac{(q-1)(q^m+1)}{2(q+1)}+\Big(\frac{(q+3)(q^{m-2\ell_0}+1)}{2(q+1)}-2\ell_1-3\Big)q^{2\ell_0}\\
				&=q^m-(2\ell_1+2)q^{2\ell_0}-\frac{(q-1)(q^{2\ell_0}-1)}{2(q+1)}<2n;
			\end{align*}
		\end{small}
		
		\noindent thus $[aq^t]_{2n}=M_t\geq \frac{(q-1)(q^m+1)}{2(q+1)}>f$.
		
		{\bf Case 3.} $2\ell_0+1\leq t\leq m-1$. Suppose $t$ is odd, then
		\begin{small}
			\begin{align*}
				aq^t&=\Big(\frac{(q-1)(q^m+1)}{2(q+1)}+\sum_{j=0}^{m-t-1}a_jq^j+\sum_{j=m-t}^{m-2\ell_0-1}a_jq^j\Big)q^t\\
				&\equiv \frac{(q+3)(q^m+1)}{2(q+1)}+\sum_{j=t}^{m-1}a_{j-t}q^j-\sum_{j=0}^{t-2\ell_0-1}a_{j+m-t}q^j\triangleq M_t~({\rm mod}~2n).
			\end{align*}
		\end{small}
		
		\noindent Direct calculations show that
		\begin{small}
			$$	M_t\geq \frac{(q+3)(q^m+1)}{2(q+1)}-\sum_{j=0}^{t-2\ell_0-1}(q-1)q^j\geq \frac{(q+3)(q^m+1)}{2(q+1)}-q^{m-2\ell_0-1}+1>0,~{\rm and}$$
			$$M_t\leq \frac{(q+3)(q^m+1)}{2(q+1)}+\sum_{j=t}^{m-1}(q-1)q^j\leq \frac{(q+3)(q^m+1)}{2(q+1)}+q^m-q^{2\ell_0+1}<4n.$$
		\end{small}
		If $0<M_t<2n$, then $[aq^t]_{2n}=M_t>f$. If $2n<M_t<4n$, then
		\begin{small}
			\begin{align*}
				[aq^t]_{2n}&=M_t-2n=\sum_{j=t}^{m-1}a_{j-t}q^j-\frac{(q-1)(q^m+1)}{2(q+1)}-\sum_{j=0}^{t-2\ell_0-1}a_{j+m-t}q^j\\
				&=\sum_{j=t}^{m-1}\Big(a_{j-t}+(-1)^{j+1}\frac{q-1}{2}\Big)q^j-\frac{(q-1)(q^t+1)}{2(q+1)}-\sum_{j=0}^{t-2\ell_0-1}a_{j+m-t}q^j\\
				&\geq q^t-\frac{(q-1)(q^t+1)}{2(q+1)}-\sum_{j=0}^{t-2\ell_0-1}a_{j+m-t}q^j=\frac{(q^{2\ell_0+1}+3q^{2\ell_0}-2q-2)q^{t-2\ell_0}+q+3}{2(q+1)}\\
				&\geq \frac{(q^{2\ell_0+1}+3q^{2\ell_0}-2q-2)q+q+3}{2(q+1)}=\frac{(q+3)(q^{2\ell_0+1}+1)}{2(q+1)}-q>f.
			\end{align*}
		\end{small}
		
		Suppose $t$ is even, then
		\begin{small}
			\begin{align*}
				aq^t&=\Big(\frac{(q-1)(q^m+1)}{2(q+1)}+\sum_{j=0}^{m-t-1}a_jq^j+\sum_{j=m-t}^{m-2\ell_0-1}a_jq^j\Big)q^t\\
				&\equiv \frac{(q-1)(q^m+1)}{2(q+1)}+\sum_{j=t}^{m-1}a_{j-t}q^j-\sum_{j=0}^{t-2\ell_0-1}a_{j+m-t}q^j\triangleq M_t~({\rm mod}~2n).
			\end{align*}
		\end{small}
		
		\noindent One can check that
		\begin{small}
			$$M_t\geq \frac{(q-1)(q^m+1)}{2(q+1)}-\sum_{j=0}^{t-2\ell_0-1}(q-1)q^j\geq \frac{(q-1)(q^m+1)}{2(q+1)}-q^{m-2\ell_0-1}+1>0, ~{\rm and}$$
			$$M_t\leq \frac{(q-1)(q^m+1)}{2(q+1)}+\sum_{j=t}^{m-1}(q-1)q^j\leq \frac{(q-1)(q^m+1)}{2(q+1)}+q^m-q^{2\ell_0+2}<4n.$$
		\end{small}
		If $0<M_t<2n$, then $[aq^t]_{2n}=M_t>f$. If $2n<M_t<4n$, then
		\begin{small}
			\begin{align*}
				[aq^t]_{2n}&=M_t-2n=\sum_{j=t}^{m-1}a_{j-t}q^j-\frac{(q+3)(q^m+1)}{2(q+1)}-\sum_{j=0}^{t-2\ell_0-1}a_{j+m-t}q^j\\
				&=\sum_{j=t}^{m-1}\Big(a_{j-t}+(-1)^{j+1}\frac{q+3}{2}\Big)q^j+\frac{(q+3)(q^t-1)}{2(q+1)}-\sum_{j=0}^{t-2\ell_0-1}a_{j+m-t}q^j.
			\end{align*}
		\end{small}
		
		\noindent  In this subcase, we assert that $\sum_{j=t}^{m-1}\big(a_{j-t}+(-1)^{j+1}\frac{q+3}{2}\big)q^j\geq 0$; otherwise, $[aq^t]_{2n}\leq -q^t+\frac{(q+3)(q^t-1)}{2(q+1)}=-\frac{(q-1)(q^t-1)}{2(q+1)}-1<0$, which is impossible. Consequently,
		\begin{small}
			\begin{align*}
				[aq^t]_{2n}&\geq \frac{(q+3)(q^t-1)}{2(q+1)}-\sum_{j=0}^{t-2\ell_0-1}(q-1)q^j=\frac{(q^{2\ell_0+3}+3q^{2\ell_0+2}-2q^3-2q^2)q^{t-2\ell_0-2}+q-1}{2(q+1)}\\
				&\geq \frac{q^{2\ell_0+3}+3q^{2\ell_0+2}-2q^3-2q^2+q-1}{2(q+1)}=\frac{(q+3)(q^{2\ell_0+2}+q)}{2(q+1)}-q^2-\frac{q+1}{2}>f.
			\end{align*}
		\end{small}
		
		{\bf Case 4.} $m\leq t\leq 2m-1$. Let $t'=t-m$, then $0\leq t'\leq m-1$ and $[aq^t]_{2n}=2n-[aq^{t'}]_{2n}$. It follows from Cases 1-3 that $[aq^t]_{2n}>f$. The proof is completed.
	\end{proof}
	
	Combining Lemmas \ref{l2.1}, \ref{L3.1} and \ref{L3.2}, we have the following result.
	\begin{Theorem}\label{T3.1}
		Suppose $q>3$ and $m\geq 3$ is odd. For $2\leq \delta<\frac{\phi_1+3}{2}$, we have
		\begin{align*}
			d(\mathcal{C}_{(q,n,-1,\delta,1)}^{\perp})\geq \begin{cases}
				\frac{I_2(\delta)-I_1(\delta)}{2}, & {\rm if}~2\leq \delta<\frac{\phi_2+3}{2};\\
				2, & {\rm if}~\frac{\phi_2+3}{2}\leq \delta<\frac{\phi_1+3}{2},
			\end{cases}
		\end{align*}
		where $I_1(\delta)$ and $I_2(\delta)$ are given by Lemmas \ref{L3.1} and \ref{L3.2}, respectively.
	\end{Theorem}
	
	\begin{Example}{\rm
			Take $q=7$ and $m=3$. By Theorem \ref{T3.1}, the lower bound on the minimum distance of $\mathcal{C}_{(q,n,-1,2,1)}^{\perp}$ is $123$. By Magma, the true minimum distance of $\mathcal{C}_{(q,n,-1,2,1)}^{\perp}$ is $138$.}
	\end{Example}
	
	To show the condition for $\mathcal{C}_{(q,n,-1,\delta,1)}$ being a negacyclic dually-BCH code, we first give a key lemma.
	\begin{Lemma}\label{L3.3}
		Suppose $q>3$ and $m\geq 3$ is odd. For $2\leq \delta< \frac{\phi_2+3}{2}$, let $M(\delta)=T^{\perp}\setminus \big(C_{I_1(\delta)+2}^{(q,2n)}\cup C_{I_1(\delta)+4}^{(q,2n)}\cup \cdots\cup C_{I_2(\delta)-2}^{(q,2n)}\big)$. Then $\frac{q^m+1}{q+1}\in M(\delta)$ if $2\leq \delta < \frac{q^m+3q+4}{2(q+1)}$ and $\phi_2\in M(\delta)$ if $\frac{q^m+3q+4}{2(q+1)}\leq \delta< \frac{\phi_2+3}{2}$.
	\end{Lemma}
	
	\begin{proof}
		When $2\leq \delta<\frac{q^m+3q+4}{2(q+1)}$, $1\leq 1+2(\delta-2)<\frac{q^m+1}{q+1}$. We see from Lemmas \ref{L3.1} and \ref{L3.2} that ${\rm min}\{I_1(\delta)\mid 2\leq \delta<\frac{q^m+3q+4}{2(q+1)}\}=q^{m-1}$ and ${\rm max}\{I_2(\delta)\mid 2\leq \delta<\frac{q^m+3q+4}{2(q+1)}\}=q^m-q^{m-1}+1$. To prove $\frac{q^m+1}{q+1}\in M(\delta)$, we need to prove that $C_{\frac{q^m+1}{q+1}}^{(q,2n)}\subseteq [\frac{q^m+1}{q+1}, q^{m-1}]\cup [q^m-q^{m-1}+1, 2n-1]$. Note that $C_{\frac{q^m+1}{q+1}}^{(q,2n)}=\big\{\frac{q^m+1}{q+1}, \frac{q^m+q}{q+1}\big\}$. The result is then established.
		
		When $\frac{q^m+3q+4}{2(q+1)}\leq \delta< \frac{\phi_2+3}{2}$, $\frac{q^m+1}{q+1}\leq 1+2(\delta-2)<\phi_2$. Obviously, $\phi_2=\frac{q-1}{2}\big(q^{m-1}-q^{m-2}-\frac{q^{m-2}+1}{q+1}\big)\in T^{\perp}$. Suppose $q>7$. It follows from Lemma \ref{L3.1} that $f_1:={\rm min}\{I_1(\delta)\mid \frac{q^m+3q+4}{2(q+1)}\leq \delta< \frac{\phi_2+3}{2}\}=\frac{(q-1)(q^m+1)}{2(q+1)}-\frac{q-3}{2}$ as $\frac{2q\Delta_{\frac{m-1}{2},q}'-q+5}{4}\leq \frac{q^m+3q+4}{2(q+1)}\leq q^{m-1}+\frac{2q\Delta_{\frac{m-1}{2},q}'-q+1}{4}$, and it follows from Lemma \ref{L3.2} that $f2:={\rm max}\{I_2(\delta)\mid \frac{q^m+3q+4}{2(q+1)}\leq \delta< \frac{\phi_2+3}{2}\}=\frac{(q-1)(q^m+1)}{2(q+1)}+\frac{q+1}{2}$ as $\frac{q^m-q^{m-1}+7q+9}{4(q+1)}\leq \frac{q^m+3q+4}{2(q+1)}\leq q^{m-1}+\frac{q^m-q^{m-1}+3q+5}{4(q+1)}$. To prove $\phi_2\in M(\delta)$, it suffices to prove that either $[\phi_2\cdot q^t]_{2n}\leq f_1$ or $[\phi_2\cdot q^t]_{2n}\geq f_2$, where $0\leq t\leq 2m-1$. We have the following five cases.
		
		{\bf Case 1.} $t=0$. It is clear that $[\phi_2\cdot q^t]_{2n}=\phi_2<f_1$.
		
		{\bf Case 2.} $t=1$. Then
		\begin{small}
			\begin{align*}
				\phi_2\cdot q^t&\equiv \frac{q-1}{2}\Big(-1-q^{m-1}-\frac{q^{m-1}+q}{q+1}\Big)\equiv q^m+1-\frac{q-1}{2}\Big(q^{m-1}+\frac{q^{m-1}+q}{q+1}+1\Big)\\
				&=\frac{q^m+1}{2}+\frac{q^{m-1}-q^2}{q+1}+1\triangleq M_t~({\rm mod}~2n).
			\end{align*}
		\end{small}
		
		\noindent 	It is easy to see that $[\phi_2\cdot q^t]_{2n}=M_t>f_2$.
		
		{\bf Case 3.} $t=2$. Then
		\begin{small}
			\begin{align*}
				\phi_2\cdot q^t&\equiv \frac{q-1}{2}\Big(-q+1-\frac{q^m+q^2}{q+1}\Big)\equiv q^m+1-\frac{q-1}{2}\Big(\frac{q^m+q^2}{q+1}+q-1\Big)\\
				&=\frac{(q+3)(q^m+1)}{2(q+1)}-\frac{q^3+1}{q+1}+q\triangleq M_t~({\rm mod}~2n).
			\end{align*}
		\end{small}
		
		\noindent 	It follows that $[\phi_2\cdot q^t]_{2n}=M_t>f_2$.
		
		{\bf Case 4.} $3\leq t\leq m-1$. If $t$ is odd, we have
		\begin{small}
			\begin{align*}
				\phi_2\cdot q^t&=\frac{q-1}{2}\Big(q^{m-1}-q^{m-2}-\frac{(q^{t-2}+1)q^{m-t}}{q+1}+\frac{q^{m-t}-1}{q+1}\Big)q^t\\
				&\equiv \frac{q-1}{2}\Big(-q^{t-1}+q^{t-2}+\frac{q^{t-2}+1}{q+1}+\frac{q^m-q^t}{q+1}\Big)\\
				&=\frac{(q-1)(q^m+1)}{2(q+1)}-(q-1)^2\cdot q^{t-2}\triangleq M_t~({\rm mod}~2n).
			\end{align*}
		\end{small}
		
		\noindent Then $[\phi_2\cdot q^t]_{2n}=M_t\leq \frac{(q-1)(q^m+1)}{2(q+1)}-q(q-1)^2<f_1$.
		
		If $t$ is even, we have 
		\begin{small}
			\begin{align*}
				\phi_2\cdot q^t&=\frac{q-1}{2}\Big(q^{m-1}-q^{m-2}-\frac{(q^{t-2}-1)q^{m-t}}{q+1}-\frac{q^{m-t}+1}{q+1}\Big)q^t\\
				&\equiv \frac{q-1}{2}\Big(-q^{t-1}+q^{t-2}+\frac{q^{t-2}-1}{q+1}-\frac{q^m+q^t}{q+1}\Big)\\
				&\equiv q^m+1-\frac{q-1}{2}\Big(\frac{q^m+q^t}{q+1}-\frac{q^{t-2}-1}{q+1}+q^{t-1}-q^{t-2}\Big)\\
				&=\frac{(q+3)(q^m+1)}{2(q+1)}-(q-1)^2\cdot q^{t-2}\triangleq M_t~({\rm mod}~2n).
			\end{align*}
		\end{small}
		
		\noindent  Then $[\phi_2\cdot q^t]_{2n}=M_t\geq \frac{(q+3)(q^m+1)}{2(q+1)}-(q-1)^2\cdot q^{m-3}=\frac{q^m+1}{2}+q^{m-2}-\frac{q^{m-3}-1}{q+1}>f_2$.
		
		{\bf Case 5.} $m\leq t\leq 2m-1$. Let $t'=t-m$, then $0\leq t'\leq m-1$ and $[\phi_2\cdot q^t]_{2n}=2n-[\phi_2\cdot q^{t'}]_{2n}$.
		It can be seen from Cases 1-4 that either $[\phi_2\cdot q^t]_{2n}<f_1$ or $[\phi_2\cdot q^t]_{2n}>f_2$.
		
		When $q=7$, it follows from Lemmas \ref{L3.1} and \ref{L3.2} that $f_1':={\rm min}\{I_1(\delta)\mid \frac{q^m+3q+4}{2(q+1)}\leq \delta< \frac{\phi_2+3}{2}\}=\frac{(q-1)(q^m+1)}{2(q+1)}-\frac{q-1}{2}$ and $f2':={\rm max}\{I_2(\delta)\mid \frac{q^m+3q+4}{2(q+1)}\leq \delta< \frac{\phi_2+3}{2}\}=\frac{(q-1)(q^m+1)}{2(q+1)}+4$. The rest of the proof is similar to the case $q>7$.
	\end{proof}
	
	\begin{Theorem}\label{T3.2}
		Suppose $q>3$ and $m\geq 3$ is odd. For $2\leq \delta < \frac{\phi_1+3}{2}$, $\mathcal{C}_{(q,n,-1,\delta,1)}$ is a negacyclic dually-BCH code if and only if
		$$\frac{\phi_2+3}{2}\leq \delta < \frac{\phi_1+3}{2}.$$
	\end{Theorem}
	\begin{proof}
		The proof is clear with the help of Lemma \ref{L3.3}.
	\end{proof}
	\begin{Example}{\rm 
			Take $q=7$ and $m=3$. By Magma, for $2\leq \delta<\frac{\phi_1+3}{2}=66$,  $\mathcal{C}_{(q,n,-1,\delta,1)}$ is a negacyclic dually-BCH code if and only if $63\leq \delta <66$, which coincides with Theorem \ref{T3.2}.}
	\end{Example}
	
	\subsection*{B.2. The subcase $m$ is even}
	In this subsetion, we consider the subcase $m$ is even. The following lemma will be needed later.
	\begin{Lemma}\label{L3.4}
		Suppose $q>3$ and $m\geq 2$ is even. For $2\leq \delta<\frac{\phi_2+3}{2}$, let $I(\delta)$ be the odd integer such that $I(\delta)\not\in T^{\perp}$ and $\{I(\delta)+2,I(\delta)+4,\cdots,\phi_1\}\subseteq T^{\perp}$. Then
		\begin{small}
			\begin{align*}
				I(\delta)=\begin{cases}
					(\frac{q^{2\ell_0-1}-q-2}{2}+2\ell_1)q^{m-2\ell_0+1}, & {\rm if}~\delta=\frac{q^{2\ell_0-1}-q+4}{4}+\ell_1~(1\leq \ell_0\leq \frac{m}{2},~1\leq \ell_1\leq \frac{q-3}{4});\\
					\frac{q^m-q^{m-2\ell_0+1}}{2}, & {\rm if}~\frac{q^{2\ell_0-1}+5}{4}\leq \delta\leq \frac{3q^{2\ell_0-1}+3}{4}~(1\leq \ell_0\leq \frac{m}{2});\\
					\frac{q^m-q^{m-2\ell_0+1}}{2}+2\ell_1, & {\rm if}~\frac{(4\ell_1-1)q^{2\ell_0-1}+7}{4}\leq \delta\leq \frac{(4\ell_1+3)q^{2\ell_0-1}+3}{4}\\
					&(1\leq \ell_0\leq \frac{m}{2},~1\leq \ell_1<\frac{q-3}{4});\\
					\frac{q^m-q^{m-2\ell_0+1}+q-3}{2}, & {\rm if}~\frac{(q-4)q^{2\ell_0-1}+7}{4}\leq \delta\leq \frac{q^{2\ell_0}-q+6}{4}~(1\leq \ell_0\leq \frac{m}{2});\\
					(\frac{q^{2\ell_0}-q}{2}+2\ell_1)q^{m-2\ell_0}, & {\rm if}~\delta=\frac{q^{2\ell_0}-q+6}{4}+\ell_1~(1\leq \ell_0<\frac{m}{2},~1\leq \ell_1\leq \frac{q-3}{4});\\
					\frac{q^m-q^{m-2\ell_0}+2}{2}+2\ell_1, & {\rm if}~\frac{(4\ell_1+1)q^{2\ell_0}+7}{4}\leq \delta\leq \frac{(4\ell_1+5)q^{2\ell_0}+3}{4}\\
					&(1\leq \ell_0<\frac{m}{2},~0\leq \ell_1<\frac{q-3}{4});\\
					\frac{q^m-q^{m-2\ell_0}+q-1}{2}, & {\rm if}~\frac{(q-2)q^{2\ell_0}+7}{4}\leq \delta\leq \frac{q^{2\ell_0+1}-q+4}{4}~(1\leq \ell_0<\frac{m}{2}).
				\end{cases}
			\end{align*}
		\end{small}
	\end{Lemma}
	\begin{proof}
		The proof is similar to that of Lemmas \ref{L3.1} or \ref{L3.2}, and thus we omit it here.
	\end{proof}
	
	By virtue of Lemmas \ref{l2.1} and \ref{L3.4}, we obtain the following lower bound.
	\begin{Theorem}\label{T3.3}
		Suppose $q>3$ and $m\geq 2$ is even. For $2\leq \delta<\frac{\phi_1+3}{2}$, we have 
		\begin{align*}
			d(\mathcal{C}_{(q,n,-1,\delta,1)}^{\perp})\geq \begin{cases}
				n-I(\delta), & {\rm if}~2\leq \delta<\frac{\phi_2+3}{2};\\
				2, & {\rm if}~\frac{\phi_2+3}{2}\leq \delta<\frac{\phi_1+3}{2},
			\end{cases}
		\end{align*}
		where $I(\delta)$ is given by Lemma \ref{L3.4}.
	\end{Theorem}
	
	\begin{Example}{\rm
			Take $q=7$ and $m=2$. By Theorem \ref{T3.3}, the lower bounds on the minimum distances of $\mathcal{C}_{(q,n,-1,2,1)}^{\perp}$ and $\mathcal{C}_{(q,n,-1,6,1)}^{\perp}$ are $18$ and $4$, respectively. By Magma, the true minimum distances of $\mathcal{C}_{(q,n,-1,2,1)}^{\perp}$ and $\mathcal{C}_{(q,n,-1,6,1)}^{\perp}$ are $19$ and $6$, respectively.}
	\end{Example}
	
	We next present the conditions for $\mathcal{C}_{(q,n,-1,\delta,1)}$ to be a negacyclic dually-BCH code. The cases $m=2$ and $m>2$ should be discussed seperately.
	\begin{Lemma}\label{L3.5}
		Suppose $q>3$ and $m=2$. For $3\leq \delta<\frac{\phi_2+3}{2}$, let  $M(\delta)=T^{\perp}\setminus \big(C_{I(\delta)+2}^{(q,2n)}\cup C_{I(\delta)+4}^{(q,2n)}\cup \cdots\cup C_{\phi_1}^{(q,2n)}\big)$. 
		
		{\rm (1)} If $q=7$, then $\phi_2 \in M(\delta)$.
		
		{\rm (2)} If $q>7$, then $2q+3\in M(\delta)$ if $3\leq \delta<q+3$ and $\phi_2\in M(\delta)$ if $q+3\leq \delta<\frac{\phi_2+3}{2}$.
	\end{Lemma}
	\begin{proof}
		The proof is routine and easy, and thus omitted here.
	\end{proof}
	
	\begin{Theorem}\label{T3.4}
		Suppose $q>3$ and $m=2$. For $2\leq \delta < \frac{\phi_1+3}{2}$, $\mathcal{C}_{(q,n,-1,\delta,1)}$ is a negacyclic dually-BCH code if and only if
		$$\delta=2~{\rm or}~\frac{\phi_2+3}{2}\leq \delta < \frac{\phi_1+3}{2}.$$
	\end{Theorem}
	
	\begin{proof}
		If $m=2$, then $n=\frac{q^2+1}{2}$. When $\delta=2$,  $T^{\perp}=(1+2\mathbb{Z}_{2n})\setminus C_1^{(q,2n)}$ and by Lemma \ref{L3.4}, $I(\delta)=q$. We claim that $T^{\perp}=C_{q+2}^{(q,2n)}\cup C_{q+4}^{(q,2n)}\cup\cdots \cup C_{\frac{q^2+1}{2}}^{(q,2n)}\cup\cdots \cup C_{q^2-q-3}^{(q,2n)}\cup C_{q^2-q-1}^{(q,2n)}$ and $\mathcal{C}_{(q,n,-1,\delta,1)}^{\perp}=\mathcal{C}_{(q,n,-1,\frac{q^2-2q+1}{2},q+2)}$ is a negacyclic BCH code with designed distance $\frac{q^2-2q+1}{2}$ with respect to $\beta$. To prove this, we need to prove that if $a\in T^{\perp}$, then $a\in C_{q+2}^{(q,2n)}\cup C_{q+4}^{(q,2n)}\cup\cdots \cup C_{q^2-q-1}^{(q,2n)}$. Note that $C_1^{(q,2n)}=\{1,q,q^2-q+1,q^2\}$ and that $C_j^{(q,2n)}=C_{2n-j}^{(q,2n)}$ for $j\in 1+2\mathbb{Z}_{2n}$, thus we just need consider odd integer $a$ with $3\leq a\leq q-2$ or $q+2\leq a\leq \frac{q^2+1}{2}$. If $q+2\leq a\leq \frac{q^2+1}{2}$, this holds trivially. If $3\leq a\leq q-2$, then $3q\leq aq\leq q^2-2q$; the result also holds. The rest of the proof is straightforward from Lemma \ref{L3.5}.
	\end{proof}	
	
	\begin{Example}{\rm 
			Take $q=7$ and $m=2$. By Magma, for $2\leq \delta<\frac{\phi_1+3}{2}=14$,  $\mathcal{C}_{(q,n,-1,\delta,1)}$ is a negacyclic dually-BCH code if and only if $\delta=2$ or $10\leq \delta <14$, which coincides with Theorem \ref{T3.4}.}
	\end{Example}
	
	For the case $m>2$, using the similar methods as before, we get the following lemma.
	\begin{Lemma}\label{L3.6}
		Suppose $q>3$ and $m>2$ is even. For $2\leq \delta<\frac{\phi_2+3}{2}$, let  $M(\delta)=T^{\perp}\setminus \big(C_{I(\delta)+2}^{(q,2n)}\cup C_{I(\delta)+4}^{(q,2n)}\cup \cdots\cup C_{\phi_1}^{(q,2n)}\big)$. Then
		
		{\rm (1)} $\frac{q^{m-1}+1}{q+1}\in M(\delta)$ if $2\leq \delta <\frac{q^{m-1}+3q+4}{2(q+1)}$;
		
		{\rm (2)} $\phi_3\in M(\delta)$ if $\frac{q^{m-1}+3q+4}{2(q+1)}\leq \delta<\frac{\phi_3+3}{2}$;
		
		{\rm (3)} $\phi_2\in M(\delta)$ if $\frac{\phi_3+3}{2}\leq \delta<\frac{\phi_2+3}{2}$.
	\end{Lemma}
	
	\begin{Theorem}\label{T3.5}
		Suppose $q>3$ and $m>2$ is even. For $2\leq \delta < \frac{\phi_1+3}{2}$, $\mathcal{C}_{(q,n,-1,\delta,1)}$ is a negacyclic dually-BCH code if and only if
		$$\frac{\phi_2+3}{2}\leq \delta < \frac{\phi_1+3}{2}.$$
	\end{Theorem}
	\begin{proof}
		The desired result follows directly from Lemma \ref{L3.6}.
	\end{proof}
	\begin{Example}{\rm 
			Take $q=7$ and $m=4$. By Magma, for $2\leq \delta<\frac{\phi_1+3}{2}=602$,  $\mathcal{C}_{(q,n,-1,\delta,1)}$ is a negacyclic dually-BCH code if and only if $430 \leq \delta <602$, which coincides with Theorem \ref{T3.5}.}
	\end{Example}
	
	\section{Conclusion}
	The main contributions of this paper are as follows.
	\begin{itemize}
		\item[(1)] Some lower bounds on the minimum distances of the dual codes of narrow-sense cyclic BCH codes of length $q^m+1$ over $\mathbb{F}_q$ were developed, where $q $ is an odd prime power (see Theorem \ref{t4.1}). Some examples showed that these lower bounds are the true minimum distances of the dual codes.
		
		\item[(2)] Sufficient and necessary conditions for the even-like subcodes of narrow-sense cyclic BCH codes of length $q^m+1$ over $\mathbb{F}_q$ being dually-BCH were given, where $q$ is an odd prime power and $m$ is odd or $m\equiv 2~({\rm mod}~4)$ (see Theorems \ref{t4.2} and \ref{t4.3}).
		
		\item[(3)] Some lower bounds on the minimum distances of the dual codes of narrow-sense negacyclic BCH codes of length $\frac{q^m+1}{2}$ over $\mathbb{F}_q$ were established, where $q\equiv 3~({\rm mod}~4)$ (see Theorems \ref{t3.1}, \ref{t3.3}, \ref{T3.1} and \ref{T3.3}). Some examples showed that these lower bounds are very close to the true minimum distances of the dual codes.
		
		\item[(4)] The concept of negacyclic dually-BCH codes was proposed. Sufficient and necessary conditions for narrow-sense negacyclic BCH codes of length $\frac{q^m+1}{2}$ over $\mathbb{F}_q$ being dually-BCH were presented, where $q\equiv 3~({\rm mod}~4)$ (see Theorems \ref{t3.2}, \ref{t3.4}, \ref{T3.2}, \ref{T3.4} and \ref{T3.5}).

	\end{itemize}
	
	The case $q$ is an odd prime power and $m\equiv 0~({\rm mod}~4)$ in part (2) and the case $q\equiv 1~({\rm mod}~4)$ in parts (3) and (4) were not considered in this paper. The interested reader is invited to solve these problems.

\end{document}